\documentclass[10pt,conference,letterpaper]{IEEEtran}
\usepackage{times,amsmath,epsfig,cite,amssymb,amsfonts,url,dsfont,multirow,hyperref}

\newtheorem{definition}{Definition}

\newtheorem{theorem}{Theorem}
\newtheorem{corollary}{Corollary}
\newtheorem{lemma}{Lemma}
\newtheorem{proposition}{Proposition}

\newtheorem{example}{Example}

\def\done{\hspace*{\fill} {{\small $\blacksquare$}}}
\def\extraspacing{\vspace{2mm}}

\def\figcapup{\vspace{-1mm}}
\def\figcapdown{\vspace{-1mm}}
\def\tblcapup{\vspace{-0mm}}
\def\tblcapdown{\vspace{-0mm}}
\def\tbldown{\vspace{-2mm}}
\def\codecapup{\vspace{-3mm}}
\def\codecapdown{\vspace{-2mm}}

\title{Differential Privacy via Wavelet Transforms}
\author{%
{Xiaokui Xiao \hspace{30mm} Guozhang Wang, \hspace{2mm} Johannes Gehrke}
\vspace{1.6mm}\\
\fontsize{10}{10}\selectfont\itshape
Nanyang Technological University \hspace{32mm} Cornell University \hspace{28mm}\\
Singapore \hspace{55mm} Ithaca, USA \hspace{15mm}\\
\fontsize{9}{9}\selectfont\ttfamily\upshape
xkxiao@ntu.edu.sg \hspace{24mm}  \{guoz, johannes\}@cs.cornell.edu \hspace{3mm}
}
\begin{document}
\maketitle
\begin{abstract}
Privacy preserving data publishing has attracted considerable research
interest in recent years. Among the existing solutions, {\em
  $\epsilon$-differential privacy} provides one of the strongest
privacy guarantees. Existing data publishing methods that achieve $\epsilon$-differential
privacy, however, offer little data utility. In particular, if the
output dataset is used to answer count queries, the noise in
the query answers can be proportional to the number of tuples in the
data, which renders the results useless.

In this paper, we develop a data publishing technique that ensures
$\epsilon$-differential privacy while providing accurate answers for
{\em range-count queries}, i.e., count queries where the predicate on
each attribute is a range. The core of our solution is a framework
that applies {\em wavelet transforms} on the data
before adding noise to it. We present instantiations of the proposed
framework for both ordinal and nominal data, and we provide a
theoretical analysis on their privacy and utility guarantees. In an
extensive experimental study on both real and synthetic data, we show
the effectiveness and efficiency of our solution.
\end{abstract}

% NOTE keywords are not used for conference papers so do not populate them

\section{Introduction} \label{sec:intro}

\noindent
\emph{The boisterous sea of liberty is never without a wave.}
--- Thomas Jefferson.
\vspace{1ex}

Numerous organizations, like census bureaus and hospitals, maintain
large collections of personal information (e.g., census data and
medical records). Such data collections are of significant research
value, and there is much benefit in making them publicly
available. Nevertheless, as the data is sensitive in nature, proper
measures must be taken to ensure that its publication does not
endanger the privacy of the individuals that contributed the data. A
canonical solution to this problem is to modify the data before
releasing it to the public, such that the modification prevents
inference of private information while retaining statistical
characteristics of the data.

A plethora of techniques have been proposed for privacy preserving
data publishing (see \cite{aw89,fwc10} for surveys). Existing
solutions make different assumptions about the background knowledge of
an adversary who would like to attack the data --- i.e., to learn the private
information about some individuals. Assumptions about the background
knowledge of the adversary determine what types of attacks are
possible \cite{wfw07,gks08,k09}. A solution that makes very
conservative assumptions about the adversary's background knowledge is
{\em $\epsilon$-differential privacy} \cite{dmn06}. Informally,
$\epsilon$-differential privacy requires that the data to be published
should be generated using a randomized algorithm $\mathcal{G}$, such
that the output of $\mathcal{G}$ is not very sensitive to any
particular tuple in the input, i.e., the output of $\mathcal{G}$
should rely mainly on general properties of the data. This ensures
that, by observing the data modified by $\mathcal{G}$, the adversary
is not able to infer much information about any individual tuple
in the input data, and hence, privacy is preserved.

The simplest method to enforce $\epsilon$-differential privacy, as
proposed by Dwork et al.\cite{dmn06}, is to first derive the frequency
distribution of the tuples in the input data, and then publish a noisy
version of the distribution. For example, given the medical records in
Table~\ref{tbl:intro-micro}, Dwork et al.'s method first maps the
records to the {\em frequency matrix} in Table~\ref{tbl:intro-freq},
where each entry in the first (second) column stores the number of
diabetes (non-diabetes) patients in Table~\ref{tbl:intro-micro} that
belong to a specific age group. After that, Dwork et al.'s method adds an
independent noise\footnote{Throughout the paper, we use the term
``noise'' to refer to a random variable with a zero mean.} with a
$\Theta(1)$ variance to each entry in Table~\ref{tbl:intro-freq} (we
will review this in detail in Section~\ref{sec:prelim-basic}), and
then publishes the noisy frequency matrix.

\begin{table}[t]
%\centering %\hspace{-3mm}
\begin{minipage}[t]{1.6 in}
%\centering
\tblcapup \caption{Medical Records} \tblcapdown  \label{tbl:intro-micro}
\begin{small}
\begin{tabular}{|c|c|}
\hline
{\bf Age} & {\bf Has Diabetes?} \\ \hline
& \\ [-0.9em]\hline
$< 30$ & No \\ \hline
$< 30$ & No \\ \hline
$30$-$39$ & No \\ \hline
$40$-$49$ & No \\ \hline
$40$-$49$ & Yes \\ \hline
$40$-$49$ & No \\ \hline
$50$-$59$ & No \\ \hline
$\ge 60$ & Yes \\ \hline
\end{tabular}
\tbldown
\end{small}
\end{minipage}
\begin{minipage}[t]{1.6 in}
%\centering
\tblcapup \caption{Frequency Matrix} \tblcapdown \label{tbl:intro-freq}
\begin{small}
\begin{tabular}{@{}r@{} c|c|c|}
%\cline{3-4}
\multicolumn{2}{c}{} & \multicolumn{2}{c}{\bf Has Diabetes?} \\ %\cline{3-4}
 & \multicolumn{1}{c}{} & \multicolumn{1}{c}{$\,\:$Yes$\,\:$} & \multicolumn{1}{c}{No}  \\ \cline{3-4}
%& & & \\ [-0.9em]\cline{2-4}
\multicolumn{1}{@{}r@{}}{\multirow{5}{*}{\bf Age}} & \multicolumn{1}{c|}{$< 30$}  & $0$ & $2$ \\ \cline{3-4}
\multicolumn{1}{@{}r@{}}{} & \multicolumn{1}{c|}{$30$-$39$} & $0$ & $1$ \\ \cline{3-4}
\multicolumn{1}{@{}r@{}}{} & \multicolumn{1}{c|}{$40$-$49$} & $1$ & $2$ \\ \cline{3-4}
\multicolumn{1}{@{}r@{}}{} & \multicolumn{1}{c|}{$50$-$59$} & $0$ & $1$ \\ \cline{3-4}
\multicolumn{1}{@{}r@{}}{} & \multicolumn{1}{c|}{$\ge 60$} & $1$ & $0$ \\ \cline{3-4}
\end{tabular}
\tbldown
\end{small}
\end{minipage}
\end{table}

Intuitively, the noisy frequency matrix preserves privacy, as it
conceals the exact data distribution. In addition, the matrix can provide approximate results for any
queries about Table~\ref{tbl:intro-micro}. For instance, if a user
wants to know the number of diabetes patients with age under $50$,
then she can obtain an approximate answer by summing up the first
three entries in the first column of the noisy frequency matrix.

\extraspacing
\noindent {\bf Motivation.}
Dwork et al.'s method provides reasonable accuracy for queries about
individual entries in the frequency matrix, as it injects only a small
noise (with a constant variance) into each entry. For aggregate
queries that involve a large number of entries, however, Dwork et
al.'s method fails to provide useful results. In particular, for a
count query answered by taking the sum of a constant fraction of the entries in the noisy
frequency matrix, the approximate query result has a $\Theta(m)$ noise
variance, where $m$ denotes the total number of entries in
the matrix. Note that $m$ is typically an enormous number, as
practical datasets often contain multiple attributes with large
domains. Hence, a $\Theta(m)$ noise variance can render the
approximate result meaningless, especially when the actual result of
the query is small.

\extraspacing
\noindent {\bf Our Contributions.}
In this paper, we introduce {\em Privelet} (\underline{pri}vacy preserving
wa\underline{velet}), a data publishing technique that not only
ensures $\epsilon$-differential privacy, but also provides accurate
results for all {\em range-count queries}, i.e., count queries where
the predicate on each attribute is a range. Specifically, {\em
Privelet} guarantees that any range-count query can be answered with a
noise whose variance is polylogarithmic in $m$. This significantly
improves over the $O(m)$ noise variance bound provided by Dwork et al.'s method.

The effectiveness of {\em Privelet} results from a novel application
of {\em wavelet transforms}, a type of linear transformations that has
been widely adopted for image processing \cite{sds96} and approximate
query processing \cite{cgr01}. As with Dwork et al.'s method, {\em
Privelet} preserves privacy by modifying the frequency matrix $M$ of
the input data. Instead of injecting noise directly into $M$, however,
{\em Privelet} first applies a wavelet transform on $M$, converting
$M$ to another matrix $C$. {\em Privelet} then adds a
polylogarithmic noise to each entry in $C$, and maps $C$ back to
a noisy frequency matrix $M^*$. The matrix $M^*$ thus obtained has an
interesting property: The result of any range-count query on $M^*$ can
be expressed as a weighted sum of a polylogarithmic number of entries
in $C$. Furthermore, each of these entries contributes at most
polylogarithmic noise variance to the weighted sum. Therefore, the
variance of the noise in the query result is bounded by a
polylogarithm of $m$.

The remainder of the paper is organized as
follows. Section~\ref{sec:def} gives a formal problem definition and
reviews Dwork et al.'s solution. In Section~\ref{sec:overview}, we
present the {\em Privelet} framework for incorporating wavelet
transforms in data publishing, and we establish a sufficient condition
for achieving $\epsilon$-differential privacy under the framework. We
then instantiate the framework with three differential wavelet
transforms. Our first instantiation in Section \ref{sec:ord} is based
on the Haar wavelet transform \cite{sds96}, and is applicable for
one-dimensional ordinal data. Our second instantiation in Section
\ref{sec:nom} is based on a novel {\em nominal wavelet transform},
which is designed for tables with a single nominal attribute. Our
third instantiation in Section \ref{sec:multi} is a composition of the
first two and can handle multi-dimensional data with both ordinal and
nominal attributes. We conduct a rigorous analysis on the properties of
each instantiation, and provide theoretical bounds on privacy and
utility guarantees, as well as time complexities. In Section \ref{sec:exp}, we demonstrate the
%jg01: Note that we cannot ``verify'' anything with experiments :-).
effectiveness and efficiency of {\em Privelet} through extensive
experiments on both real and synthetic data. Section~\ref{sec:related} discusses related work. In
Section~\ref{sec:conclu}, we conclude with directions for future work.

\section{Preliminaries} \label{sec:def}

\subsection{Problem Definition} \label{sec:prelim-def}

Consider that we want to publish a relational table $T$ that contains
$d$ attributes $A_1$, $A_2$, ..., $A_d$, each of which is either {\em
  ordinal} (i.e., discrete and ordered) or {\em nominal} (i.e.,
discrete and unordered). Following previous work \cite{i02,gkk07}, we
assume that each nominal attribute $A_i$ in $T$ has an associated {\em
  hierarchy}, which is a tree where (i) each leaf is a value in the
domain of $A_i$, and (ii) each internal node summarizes the leaves in
its subtree. Figure~\ref{fig:prelim-hierarchy} shows an example
hierarchy of countries. We define $n$ as the number of tuples in $T$,
and $m$ as the size of the multi-dimensional domain on which $T$ is
defined, i.e., $m = \prod_{i=1}^{d}|A_i|$.

\begin{figure}[t]
\centering
\includegraphics[width=78mm]{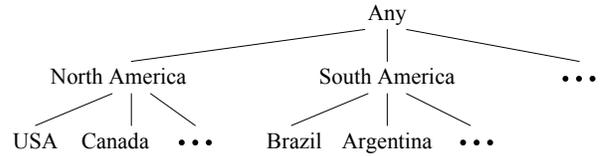}
\figcapup
\caption{A Hierarchy of Countries}
\figcapdown
\label{fig:prelim-hierarchy}
\end{figure}

We aim to release $T$ using
an algorithm that ensures $\epsilon$-differential privacy.
\extraspacing
\begin{definition}[$\epsilon$-Differential Privacy \cite{dmn06}] \label{def:prelim-epsilon}
A randomized algorithm $\mathcal{G}$ satisfies $\epsilon$-differential privacy,
if and only if (i) for any two tables $T_1$ and $T_2$ that differ only in
one tuple, and (ii) for any output $O$ of $\mathcal{G}$, we have
\begin{equation*} %\label{eqn:prelim-epsilon}
\; \quad \qquad Pr\left\{\mathcal{G}(T_1) = O\right\} \le e^{\epsilon} \cdot Pr\left\{\mathcal{G}(T_2) = O\right\}. \qquad \quad \; \blacksquare
\end{equation*}
\end{definition}
\extraspacing

We optimize the utility of the released data for OLAP-style
{\em range-count queries} in the following form:
\begin{tabbing}
\hspace{6mm} \=  \kill
\noindent \> \texttt{SELECT} \texttt{COUNT}(*) \texttt{FROM} $T$ \\
\noindent \> \texttt{WHERE} $A_1 \in S_1$ \texttt{AND} $A_2 \in S_2$ \texttt{AND} ...  \texttt{AND} $A_d \in S_d$
\end{tabbing}
For each ordinal attribute $A_i$, $S_i$ is an interval defined on the
domain of $A_i$. If $A_i$ is nominal, $S_i$ is a set that contains
either (i) a leaf in the hierarchy of $A_i$ or (ii) all leaves in the subtree of an internal node
in the hierarchy of $A_i$ --- this
is standard for OLAP-style navigation using roll-up or drill-down.
For example, given the hierarchy in Figure~\ref{fig:prelim-hierarchy},
examples of $S_i$ are $\{USA\}$, $\{Canada\}$, and the set of all countries in
North America, etc. Range-count queries are essential for various
analytical tasks, e.g., OLAP, association rule mining and decision tree construction over a data cube.

\subsection{Previous Approaches} \label{sec:prelim-basic}

As demonstrated in Section~\ref{sec:intro}, the information in $T$ can
be represented by a $d$-dimensional {\em frequency matrix} $M$ with
$m$ entries, such that (i) the $i$-th ($i \in [1, d]$) dimension of
$M$ is indexed by the values of $A_i$, and (ii) the entry in $M$ with
a coordinate vector $\langle x_1, x_2, \ldots, x_d \rangle$ stores the
number of tuples $t$ in $T$ such that $t = \langle x_1, x_2, \ldots,
x_d \rangle$. (This is the lowest level of the data
cube of $T$.) Observe that any range-count query on $T$ can be
answered using $M$, by summing up the entries in $M$ whose coordinates
satisfy all query predicates.

Dwork et al.\ \cite{dmn06} prove that $M$ can be released in a privacy
preserving manner by adding a small amount of noise to each entry in
$M$ independently. Specifically, if the noise $\eta$ follows a {\em
Laplace distribution} with a probability density function
\begin{equation} \label{eqn:basic-laplace}
Pr\{\eta = x\} = \frac{1}{2\lambda} e^{-|x|/\lambda},
\end{equation}
then the noisy frequency matrix ensures $(2/\lambda)$-differential
privacy. We refer to $\lambda$ as the {\em magnitude} of the
noise. Note that a Laplace noise with magnitude $\lambda$ has a
variance $2 \lambda^2$.

\extraspacing \noindent
{\bf Privacy Analysis.} To explain why Dwork et al.'s method ensures
privacy, suppose that we arbitrarily modify a tuple in $T$. In that
case, the frequency matrix of $T$ will change in exactly two entries,
each of which will be decreased or increased by one. For example,
assume that we modify the first tuple in Table~\ref{tbl:intro-micro},
by setting its age value to ``$30$-$39$''. Then, in the frequency
matrix in Table~\ref{tbl:intro-freq}, the first (second) entry of the
second column will be decreased (increased) by one. Intuitively, such
small changes in the entries can be easily offset by the noise added
to the frequency matrix. In other words, the noisy matrix is
insensitive to any modification to a single tuple in $T$. Thus, it is
difficult for an adversary to infer private information from the noisy
matrix. More formally, Dwork et al.'s method is based on the concept
of {\em sensitivity}.

\extraspacing
\begin{definition} [Sensitivity \cite{dmn06}]
Let $F$ be a set of functions, such that the output of each function
$f \in F$ is a real number. The sensitivity of $F$ is defined as
\begin{equation}
S(F) = \max_{T_1, T_2} \sum_{f \in F} \left|f(T_1) - f(T_2)\right|,
\end{equation}
where $T_1$ and $T_2$ are any two tables that differ in only one tuple. \done
\end{definition}
\extraspacing

Note that the frequency matrix $M$ of $T$ can be regarded as the
outputs of a set of functions, such that each function maps $T$ to an
entry in $M$. Modifying any tuple in $T$ will only change the values
of two entries (in $M$) by one. Therefore, the set of functions
corresponding to $M$ has a sensitivity of $2$. The following theorem
shows a sufficient condition for $\epsilon$-differential privacy.

\extraspacing
\begin{theorem}[\cite{dmn06}] \label{thrm:prelim-basic-epsilon}
Let $F$ be a set of functions with a sensitivity $S(F)$. Let
$\mathcal{G}$ be an algorithm that adds independent noise to the
output of each function in $F$, such that the noise follows a Laplace
distribution with magnitude $\lambda$. Then, $\mathcal{G}$ satisfies
$(S(F)/\lambda)$-differential privacy. \done
\end{theorem}
\extraspacing

By Theorem~\ref{thrm:prelim-basic-epsilon}, Dwork et al.'s method
guarantees $(2 / \lambda)$-differential privacy, since $M$ corresponds
to a set of queries on $T$ with a sensitivity of $2$.

\extraspacing \noindent
{\bf Utility Analysis.} Suppose that we answer a range-count query
using a noisy frequency matrix $M^*$ generated by Dwork et al.'s
method. The noise in the query result has a variance $\Theta(m /
\epsilon^2)$ in the worst case. This is because (i) each entry in
$M^*$ has a noise variance $8 / \epsilon^2$ (by
Equation~\ref{eqn:basic-laplace} and $\epsilon = 2/\lambda$), and (ii)
a range-count query may cover up to $m$ entries in $M^*$. Therefore, although
Dwork et al.'s method provides reasonable accuracy for queries that involve a small number of
entries in $M^*$, it offers unsatisfactory utility for large queries that cover many
entries in $M^*$.

\section{The Privelet Framework} \label{sec:overview}

This section presents an overview of our {\em Privelet} technique. We
first clarify the key steps of {\em Privelet} in
Section~\ref{sec:overview-steps}, and then provide in
Section~\ref{sec:overview-privacy} a sufficient condition for
achieving $\epsilon$-differential privacy with {\em Privelet}.

\subsection{Overview of Privelet} \label{sec:overview-steps}

%jg01: We don't want to look too incremental
%As with Dwork et al.'s method,
Our {\em Privelet} technique takes as input a relational table $T$ and
a parameter $\lambda$ and outputs a noisy version $M^*$ of the
frequency matrix $M$ of $T$. At a high level, {\em Privelet} works in
three steps as follows.

First, it applies a {\em wavelet transform} on
$M$. Generally speaking, a wavelet transform is an invertible linear
function, i.e., it maps $M$ to another matrix $C$, such that (i) each
entry in $C$ is a linear combination of the entries in $M$, and (ii)
$M$ can be losslessly reconstructed from $C$. The entries in $C$ are
referred to as the {\em wavelet coefficients}. Note that wavelet transforms are traditionally only defined for ordinal data,
and we create a special extension for nominal data in our setting.

Second, {\em Privelet} adds an independent Laplace noise
to each wavelet coefficient in a way that ensures
$\epsilon$-differential privacy. This results in a new matrix $C^*$
with noisy coefficients.  In the third step, {\em Privelet} (optionally) refines $C^*$, and then
maps $C^*$ back to a noisy frequency matrix $M^*$, which is returned
as the output. The refinement of $C^*$ may arbitrarily modify $C^*$,
but it does not utilize any information from $T$ or $M$. In other words, the third step of
{\em Privelet} depends only on $C^*$. This ensures that {\em Privelet} does not leak any information of $T$,
except for what has been disclosed in $C^*$. Our
solution in Section~\ref{sec:nom} incorporates a refinement procedure
to achieve better utility for range-count queries.

\subsection{Privacy Condition} \label{sec:overview-privacy}

The privacy guarantee of {\em Privelet} relies on its second step,
where it injects Laplace noise into the wavelet coefficient matrix
$C$. To understand why this achieves $\epsilon$-differential privacy,
recall that, even if we arbitrarily replace one tuple in the input
data, only two entries in the frequency matrix $M$ will be altered. In
addition, each of those two entries will be offset by exactly
one. This will incur only linear changes in the wavelet coefficients
in $C$, since each coefficient is a linear combination of the entries
in $M$. Intuitively, such linear changes can be concealed, as long as
an appropriate amount of noise is added to $C$.

In general, the noise required for each wavelet coefficient varies, as
each coefficient reacts differently to changes in $M$. {\em Privelet}
decides the amount of noise for each coefficient based on a {\em
  weight function} $\mathcal{W}$, which maps each coefficient to a
positive real number. In particular, the magnitude of the noise for a
coefficient $c$ is always set to $\lambda / \mathcal{W}(c)$, i.e., a
larger weight leads to a smaller noise. To analyze the privacy
implication of such a noise injection scheme, we introduce the concept
of {\em generalized sensitivity}.

\extraspacing
\begin{definition} [Generalized Sensitivity] \label{def:over-amplify}
Let $F$ be a set of functions, each of which takes as input a matrix
and outputs a real number. Let $\mathcal{W}$ be a function that
assigns a weight to each function $f \in F$. The generalized
sensitivity of $F$ with respect to $\mathcal{W}$ is defined as the
smallest number $\rho$ such that
\begin{equation} \nonumber
\sum_{f \in F} \Big( \mathcal{W}(f) \cdot \left|f(M) - f(M')\right| \Big) \le \rho \cdot \left\|M - M'\right\|_1,
\end{equation}
where $M$ and $M'$ are any two matrices that differ in only one entry,
and $\|M - M'\|_1 = \sum_{v \in M - M'} |v|$ is the $L_1$ distance
between $M$ and $M'$. \done
\end{definition}
\extraspacing

Observe that each wavelet coefficient $c$ can be regarded as the
output of a function $f$ that maps the frequency matrix $M$ to a real
number. Thus, the wavelet transform can be regarded as the set of functions corresponding
to the wavelet coefficients. The weight
$\mathcal{W}(c)$ we assign to each coefficient $c$ can be thought of
as a weight given to the function associated with $c$. Intuitively, the
generalized sensitivity captures the ``weighted'' sensitivity of the
wavelet coefficients with respect to changes in $M$. The following
lemma establishes the connection between generalized sensitivity and
$\epsilon$-differential privacy.

\extraspacing
\begin{lemma} \label{lmm:over-privacy}
Let $F$ be a set of functions that has a generalized sensitivity
$\rho$ with respect to a weight function $\mathcal{W}$. Let
$\mathcal{G}$ be a randomized algorithm that takes as input a table
$T$ and outputs a set $\{f(M) + \eta(f) \mid f \in F\}$ of real
numbers, where $M$ is the frequency matrix of $T$, and $\eta(f)$ is a
random variable that follows a Laplace distribution with magnitude
$\lambda / \mathcal{W}(f)$. Then, $\mathcal{G}$ satisfies $(2
\rho/\lambda)$-differential privacy.
\end{lemma}
\begin{proof}
See Appendix~\ref{sec:proof-over-privacy}.
\end{proof}
\extraspacing

By Lemma~\ref{lmm:over-privacy}, if a wavelet transform has a
generalized sensitivity $\rho$ with respect to weight function
$\mathcal{W}$, then we can achieve $\epsilon$-differential privacy by
adding to each wavelet coefficient $c$ a Laplace noise with magnitude $2 \rho / \mathcal{W}(c)$. This justifies the noise
injection scheme of {\em Privelet}.

\section{Privelet for One-Dimensional Ordinal Data} \label{sec:ord}

This section instantiates the {\em Privelet} framework with the
one-dimensional {\em Haar wavelet transform} \cite{sds96} (HWT), a
popular technique for processing one-dimensional ordinal data. The
one-dimensional HWT requires the input to be a
vector that contains $2^l$ ($l \in \mathds{N}$) totally ordered
elements. Accordingly, we assume wlog.~that (i) the frequency matrix
$M$ has a single ordinal dimension, and (ii) the number $m$ of entries
in $M$ equals $2^l$ (this can be ensured by inserting dummy values
into $M$ \cite{sds96}). We first explain the HWT in
Section~\ref{sec:ord-haar}, and then present the instantiation of {\em
  Privelet} in Section~\ref{sec:ord-prive}.

\subsection{One-Dimensional Haar Wavelet Transform} \label{sec:ord-haar}

The HWT converts $M$ into $2^l$ wavelet coefficients as
follows. First, it constructs a full binary tree $R$ with $2^l$
leaves, such that the $i$-th leaf of $R$ equals the $i$-th entry in
$M$ ($i \in [1, 2^l]$). It then generates a wavelet coefficient $c$
for each internal node $N$ in $R$, such that $c = (a_1 - a_2)/2$,
where $a_1$ ($a_2$) is the average value of the leaves in the left
(right) subtree of $N$. After all internal nodes in $R$ are processed,
an additional coefficient (referred to as the {\em base coefficient})
is produced by taking the mean of all leaves in $R$.  For convenience,
we refer to $R$ as the {\em decomposition tree} of $M$, and slightly
abuse notation by not distinguishing between an internal node in $R$
and the wavelet coefficient generated for the node.

\begin{figure}[t]
\centering
\includegraphics[width=45mm]{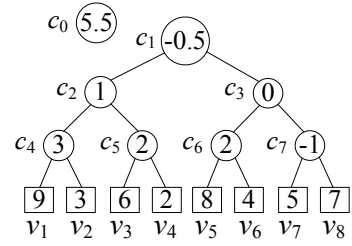}
\figcapup
\caption{One-Dimensional Haar Wavelet Transform}
\figcapdown
\label{fig:ord-haar}
\end{figure}

\extraspacing
\begin{example}
Figure~\ref{fig:ord-haar} illustrates an HWT on a
one-dimensional frequency matrix $M$ with $8$ entries $v_1, \ldots,
v_8$. Each number in a circle (square) shows the value of a wavelet
coefficient (an entry in $M$). The base coefficient $c_0$ equals the
mean $5.5$ of the entries in $M$. The coefficient $c_1$ has a value
$-0.5$, because (i) the average value of the leaves in its left
(right) subtree equals $5$ ($6$), and (ii) $(5 - 6)/2 = -0.5$. \done
\end{example}
\extraspacing

Given the Haar wavelet coefficients of $M$, any entry $v$ in $M$ can
be easily reconstructed. Let $c_0$ be the base coefficient, and $c_i$
($i \in [1, l]$) be the ancestor of $v$ at level $i$ of the
decomposition tree $R$ (we regard the root of $R$ as level $1$). We
have
\begin{equation} \label{eqn:ord-recon}
v = c_0 + \sum_{i=1}^l \left(g_i \cdot c_i\right),
\end{equation}
where $g_i$ equals $1$ ($-1$) if $v$ is in the left (right) subtree of $c_i$.

\extraspacing
\begin{example}
In the decomposition tree in Figure~\ref{fig:ord-haar}, the leaf $v_2$
has three ancestors $c_1=-0.5$, $c_2=1$, and $c_4=3$. Note that $v_2$
is in the right (left) subtree of $c_4$ ($c_1$ and $c_2$), and the
base coefficient $c_0$ equals $5.5$. We have $v_2 = 3 = c_0 + c_1 +
c_2 - c_4$. \done
\end{example}
%\extraspacing

\subsection{Instantiation of Privelet} \label{sec:ord-prive}

{\em Privelet} with the one-dimensional HWT follows the three-step
paradigm introduced in Section~\ref{sec:overview-steps}. Given a
parameter $\lambda$ and a table $T$ with a single ordinal attribute,
{\em Privelet} first computes the Haar wavelet coefficients of the
frequency matrix $M$ of $T$.  It then adds to each coefficient $c$
a random Laplace noise with magnitude $\lambda / \mathcal{W}_{Haar}(c)$,
where $\mathcal{W}_{Haar}$ is a weight function defined as follows:
For the base coefficient $c$, $\mathcal{W}_{Haar}(c)=m$; for a
coefficient $c_i$ at level $i$ of the decomposition tree,
$\mathcal{W}_{Haar}(c_i) = 2^{l-i+1}$. For example, given the wavelet
coefficients in Figure~\ref{fig:ord-haar}, $\mathcal{W}_{Haar}$ would
assign weights $8$, $8$, $4$, $2$ to $c_0$, $c_1$, $c_2$, and $c_4$,
respectively. After the noisy wavelet coefficients are computed, {\em
  Privelet} converts them back to a noisy frequency matrix $M^*$ based
on Equation~\ref{eqn:ord-recon}, and then terminates by returning
$M^*$.

%\extraspacing \noindent
%{\bf Privacy Analysis.}
This instantiation of {\em Privelet} with the one-dimensional HWT has the
following property.  \extraspacing
\begin{lemma} \label{lmm:ord-amp}
The one-dimensional HWT has a generalized sensitivity of $1 + \log_2
m$ with respect to the weight function $\mathcal{W}_{Haar}$.
\end{lemma}
\begin{proof}
Let $C$ be the set of Haar wavelet coefficients of the input matrix
$M$. Observe that, if we increase or decrease any entry $v$ in $M$ by
a constant $\delta$, only $1 + \log_2 m$ coefficients in $C$ will be
changed, namely, the base coefficient $c_0$ and all ancestors of $v$
in the decomposition tree $R$. In particular, $c_0$ will be offset by
$\delta / m$; for any other coefficient, if it is at level $i$ of the decomposition tree $R$, then it will change
by $\delta / 2^{l-i+1}$. Recall that $\mathcal{W}_{Haar}$ assigns a
weight of $m$ to $c_0$, and a weight of $2^{l-i+1}$ to any coefficient
at level $i$ of $R$. Thus, the generalized sensitivity of the
one-dimensional Haar wavelet transform with respect to
$\mathcal{W}_{Haar}$ is
\begin{equation*}
\Big(m \cdot \delta / m + \sum_{i=1}^{l}\big(2^{l-i+1} \cdot 2 \cdot
\delta / 2^{l-i+1}\big) \Big) / \delta \; = \; 1 + \log_2m.
\end{equation*}
\end{proof}
\extraspacing

By Lemmas \ref{lmm:over-privacy} and \ref{lmm:ord-amp}, {\em Privelet}
with the one-dimensional HWT ensures
$\epsilon$-differential privacy with $\epsilon = 2 (1 + \log_2 m) /
\lambda$, where $\lambda$ is the input parameter. On the other hand,
       {\em Privelet} also provides strong utility guarantee for
       range-count queries, as shown in the following lemma.

\extraspacing
\begin{lemma} \label{lmm:ord-utility}
Let $C$ be a set of one-dimensional Haar wavelet coefficients such
that each coefficient $c \in C$ is injected independent noise with a
variance at most $\left(\sigma / \mathcal{W}_{Haar}(c)\right)^2$. Let
$M^*$ be the noisy frequency matrix reconstructed from $C$. For any
range-count query answered using $M^*$, the variance of noise in the
answer is at most $(2 + \log_2 |M^*|)/2 \cdot \sigma^2$.
\end{lemma}
\begin{proof}
See Appendix~\ref{sec:proof-ord-utility}.
\end{proof}
\extraspacing

By Lemmas \ref{lmm:ord-amp} and \ref{lmm:ord-utility}, {\em Privelet}
achieves $\epsilon$-differential privacy while ensuring that the
result of any range-count query has a noise variance bounded by
\begin{equation} \label{eqn:ord-bound}
(2 + \log_2 m) \cdot (2 + 2 \log_2 m)^2 / \epsilon^2 \; = \; O\big((\log_2 m)^3/\epsilon^2\big)
\end{equation}
In contrast, as discussed in Section~\ref{sec:prelim-basic}, with the
same privacy requirement, Dwork et al.'s method incurs a noise
variance of $O(m/\epsilon^2)$ in the query answers.

Before closing this section, we point out that {\em Privelet} with the
one-dimensional HWT has an $O(n + m)$ time
complexity for construction. This follows from the facts that (i)
mapping $T$ to $M$ takes $O(m+n)$ time, (ii) converting $M$ to and
from the Haar wavelet coefficients incur $O(m)$ overhead \cite{sds96},
and (iii) adding Laplace noise to the coefficients takes $O(m)$ time.

\section{Privelet for One-dimensional Nominal Data} \label{sec:nom}

This section extends {\em Privelet} for one-dimensional nominal data
by adopting a novel {\em nominal wavelet
  transform}. Section~\ref{sec:nom-wave}
introduces the new transform, and Section~\ref{sec:nom-noise} explains
the noise injection scheme for nominal wavelet
coefficients. Section~\ref{sec:nom-bounds} analyzes the privacy and
utility guarantees of the algorithm and its time
complexity. Section~\ref{sec:nom-compare} compares the algorithm with
an alternative solution that employs the HWT.

\subsection{Nominal Wavelet Transform} \label{sec:nom-wave}

Existing wavelet transforms are only designed for ordinal data, i.e.,
they require that each dimension of the input matrix needs to have a
totally ordered domain. Hence, they are not directly applicable on
nominal data, since the values of a nominal attribute $A$ are not
totally ordered.  One way to circumvent this issue is to impose an
artificial total order on the domain of $A$, such that for any
internal node $N$ in the hierarchy of $A$, the set of leaves in the
subtree of $N$ constitutes a contiguous sequence in the total order.

\begin{figure}[t]
\centering
\includegraphics[width=85mm]{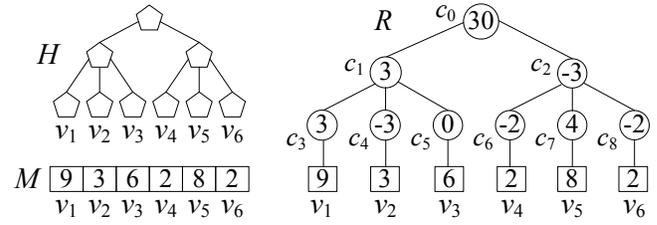}
\figcapup \caption{A nominal wavelet transform} \figcapdown
\label{fig:nom-wave}
\end{figure}

For example, given a nominal attribute $A$ with the hierarchy $H$ in
Figure~\ref{fig:nom-wave}, we impose on $A$ a total order $v_1 < v_2 <
\ldots < v_6$. As such, $A$ is transformed into an ordinal attribute
$A'$. Recall that for a nominal attribute, the range-count query
predicate ``$A \in S$'' has a special structure: $S$ either contains (i)
a leaf in the hierarchy of $A$ or (ii) all leaves in the subtree of an
internal node in the hierarchy of $A$. Therefore, $S$ is always a contiguous
range in the imposed total order of $A$.
With this transformation, we can apply {\em Privelet} with the HWT on any
one-dimensional nominal data. The noise variance bound thus obtained
is $O\big((\log_2 m)^3 / \epsilon^2\big)$ (see
Equation~\ref{eqn:ord-bound}).

While using the HWT is one possible solution, {\em Privelet} does not
stop here. We will show how to improve the above $O\big((\log_2 m)^3 /
\epsilon^2\big)$ bound to $O(h^2/\epsilon^2)$, where $h$ is the height
of the hierarchy $H$ on the nominal data. (Note that $h \le \log_2 m$
holds for any hierarchy where each internal node has at least two
children.) This improvement can result in a reduction of noise
variance by an order of magnitude or more in practice, as we will
discuss in Section \ref{sec:nom-compare}. The core of our solution
is a novel wavelet transform that creates a different
decomposition tree for generating wavelet coefficients.

A first thought for a different decomposition tree might be to use the hierarchy $H$, i.e., to generate wavelet coefficients from each internal node $N$ in $H$. Intuitively, if $N$ has only two children, then we may produce a coefficient $c$ from $N$ as in the HWT, i.e., we first compute the average value $a_1$ ($a_2$) of the leaves in the left (right) subtree of $N$, and then set $c = (a_1 + a_2)/2$. But what if $N$ has $k$ ($k > 2$) children? Should we generate one coefficient from each pair of
subtrees of $N$? But that will result in $k \choose 2$ coefficients, which is undesirable when $k$ is large. Is it possible to generate coefficients without relying on pairwise comparison of subtrees? We answer this question positively with the introduction of the {\em nominal wavelet transform}.

Given a one-dimensional frequency matrix $M$ and a hierarchy $H$ on
the entries in $M$, the nominal wavelet transform first constructs a
decomposition tree $R$ from $H$ by attaching a child node $N_c$ to
each leaf node $N$ in $H$. The value of $N_c$ is set to the value of
the entry in $M$ that corresponds to $N$. For example, given the
hierarchy $H$ in the left hand side of Figure~\ref{fig:nom-wave}, the
decomposition tree $R$ constructed from $H$ is as in right hand side
of the figure. In the second step, the nominal wavelet transform
computes a wavelet coefficient for each internal node of $R$ as
follows. The coefficient for the root node (referred to as the base
coefficient) is set to the sum of all leaves in its subtree, also
called the \emph{leaf-sum} of the node. For any other internal node,
its coefficient equals its leaf-sum minus the average leaf-sum of its
parent's children.

Given these nominal wavelet coefficients of $M$, each entry $v$
in $M$ can be reconstructed using the ancestors of $v$ in the
decomposition tree $R$. In particular, let $c_i$ be the ancestor of
$v$ in the $(i+1)$-th level of $R$, and $f_i$ be the fanout of $c_i$,
we have
\begin{equation} \label{eqn:nom-recon}
v = c_{h-1} + \sum_{i=0}^{h-2} \left( c_i \cdot \prod_{j=i}^{h-2} \frac{1}{f_j} \right),
\end{equation}
where $h$ is the height of the hierarchy $H$ on $M$.  To understand
Equation~\ref{eqn:nom-recon}, recall that $c_0$ equals the leaf-sum of
the root in $R$, while $c_k$ ($k \in [1, h-1]$) equals the leaf-sum of
$c_k$ minus the average leaf-sum of $c_{k-1}$'s
children. Thus, the leaf-sum of $c_1$ equals $c_1 + c_0 /
f_0$, the leaf-sum of $c_2$ equals $c_2 + (c_1 + c_0 / f_0) / f_1$,
and so on. It can be verified that the leaf-sum of $c_{h-1}$ equals
exactly the right hand side of Equation~\ref{eqn:nom-recon}. Since $v$
is the only leaf of $c_{h-1}$ in $R$, Equation~\ref{eqn:nom-recon}
holds.

\extraspacing
\begin{example}
Figure~\ref{fig:nom-wave} illustrates a one-dimensional frequency
matrix $M$, a hierarchy $H$ associated with $M$, and a nominal wavelet
transform on $M$. The base coefficient $c_0 = 30$ equals the sum of
all leaves in the decomposition tree. The coefficient $c_1$ equals
$3$, because (i) it has a leaf-sum $18$, (ii) the average leaf-sum of
its parent's children equals $15$, and (iii) $18 - 15 = 3$.

In the decomposition tree in Figure~\ref{fig:nom-wave}, the entry
$v_1$ has three ancestors, namely, $c_0$, $c_1$, and $c_3$, which are
at levels $1$, $2$, and $3$ of decomposition tree,
respectively. Furthermore, the fanout of $c_0$ and $c_1$ equal $2$ and
$3$, respectively. We have $v_1 = 9 = c_3 + c_0/2/3 + c_1 /3$. \done
\end{example}
\extraspacing

Note that our novel nominal wavelet transform is {\em over-complete}:
The number $m'$ of wavelet coefficients we generate is larger
than the number $m$ of entries in the input frequency matrix $M$. In
particular, $m' - m$ equals the number of internal nodes in the
hierarchy $H$ on $M$. The overhead incurred by such over-completeness,
however, is usually negligible, as the number of internal nodes in a
practical hierarchy $H$ is usually small compared to the number of
leaves in $H$.

\subsection{Instantiation of Privelet} \label{sec:nom-noise}

We are now ready to instantiate {\em Privelet} for one-dimensional
nominal data. Given a parameter $\lambda$ and a table $T$ with a
single nominal attribute, we first apply the nominal wavelet transform
on the frequency matrix $M$ of $T$. After that, we inject into each
nominal wavelet coefficient $c$ a Laplace noise with magnitude
$\lambda / \mathcal{W}_{Nom}(c)$. Specifically, $\mathcal{W}_{Nom}(c)
= 1$ if $c$ is the base coefficient, otherwise $\mathcal{W}_{Nom}(c) =
f/(2f-2)$, where $f$ is the fanout of $c$'s parent in the
decomposition tree.

Before converting the wavelet coefficients back to a noisy frequency
matrix, we refine the coefficients with a {\em mean subtraction}
procedure. In particular, we first divide all but the base
coefficients into disjoint {\em sibling groups}, such that each group
is a maximal set of noisy coefficients that have the same parent in
the decomposition tree. For example, the wavelet coefficients in
Figure~\ref{fig:nom-wave} can be divided into three sibling groups:
$\{c_1, c_2\}$, $\{c_3, c_4, c_5\}$, and $\{c_6, c_7, c_8\}$. After
that, for each sibling group, the coefficient mean is computed and
then subtracted from each coefficient in the group. Finally, we
reconstruct a noisy frequency matrix $M^*$ from the modified wavelet
coefficients (based on Equation~\ref{eqn:nom-recon}), and return $M^*$
as the output.

The mean subtraction procedure is essential to the utility guarantee
of {\em Privelet} that we will prove in
Section~\ref{sec:nom-noise}. The intuition is that, after
the mean subtraction procedure, all noisy coefficients in the same
sibling group sum up to zero; as such, for any non-root node $N$ in
the decomposition tree, the noisy coefficient corresponding to $N$
still equals the noisy leaf-sum of $N$ minus the average leaf-sum of
the children of $N$'s parent; in turn, this ensures that the
reconstruction of $M^*$ based on Equation~\ref{eqn:nom-recon} is
meaningful.

We emphasize that the mean subtraction procedure does not rely on any
information in $T$ or $M$; instead, it is performed based only on the
noisy wavelet coefficients. Therefore, the privacy guarantee of $M^*$
depends only on the noisy coefficients generated before the mean
subtraction procedure, as discussed in Section~\ref{sec:overview}.

\subsection{Theoretical Analysis} \label{sec:nom-bounds}

To prove the privacy guarantee of {\em Privelet} with the nominal
wavelet transform, we first establish the generalized sensitivity of
the nominal wavelet transform with respect to the weight function
$\mathcal{W}_{Nom}$ used in the noise injection step.

\extraspacing
\begin{lemma} \label{lmm:nom-amp}
The nominal wavelet transform has a generalized sensitivity of $h$
with respect to $\mathcal{W}_{Nom}$, where $h$ the height of the
hierarchy associated with the input frequency matrix.
\end{lemma}
\begin{proof}
Suppose that we offset an arbitrary entry $v$ in the input frequency matrix $M$ by a constant $\delta$. Then, the base coefficient of $M$ will change by $\delta$. Meanwhile, for the coefficients at level $i$ ($i \in [2, h]$) of the decomposition tree, only the sibling group $G_i$ that contains an ancestor of $v$ will be affected. In particular, the ancestor of $v$ in $G_i$ will be offset by $\delta - \delta/|G_i|$, while the other coefficients in $G_i$ will change by $\delta/|G_i|$. Recall that $\mathcal{W}_{Nom}$ assigns a weight $1$ to the base coefficient, and a weight $1/(2 - 2/|G_i|)$ for all coefficients in $G_i$. Therefore, the generalized sensitivity of the nominal wavelet transform with respect to $\mathcal{W}_{Nom}$ should be
\begin{equation*}
1 + \sum_{i=2}^{h} \Big( \frac{1}{2 - 2/|G_i|} \cdot \big(1 - \frac{1}{|G_i|} + \frac{|G_i| - 1}{|G_i|} \big) \Big) \;\; = \;\; h.
\end{equation*}
\end{proof}
\extraspacing

By Lemmas \ref{lmm:over-privacy} and \ref{lmm:nom-amp}, given a
one-dimensional nominal table $T$ and a parameter $\lambda$, {\em
  Privelet} with the nominal wavelet transform ensures
$(2h/\lambda)$-differential privacy, where $h$ is the height of the
hierarchy associated with $T$.

\extraspacing
\begin{lemma} \label{lmm:nom-utility}
Let $C'$ be a set of nominal wavelet coefficients such that each $c'
\in C'$ contains independent noise with a variance at most
$\left(\sigma/\mathcal{W}_{Nom}(c')\right)^2$. Let $C^*$ be a set of
wavelet coefficients obtained by applying a mean subtraction procedure
on $C'$, and $M^*$ be the noisy frequency matrix reconstructed from
$C^*$. For any range-count query answered using $M^*$, the variance of
the noise in the answer is less than $4 \sigma^2$.
\end{lemma}
\begin{proof}
See Appendix~\ref{sec:proof-nom-utility}.
\end{proof}
\extraspacing

By Lemmas \ref{lmm:nom-amp} and \ref{lmm:nom-utility}, when achieving
$\epsilon$-differential privacy, {\em Privelet} with the nominal
wavelet transform guarantees that each range-count query result has a
noise variance at most
\begin{eqnarray} \label{eqn:nom-bound}
4 \cdot 2 \cdot (2h)^2 / \epsilon^2 &=& O\big(h^2/\epsilon^2\big).
\end{eqnarray}
As $h \le \log_2 m$ holds in practice, the above
$O\big(h^2/\epsilon^2\big)$ bound significantly improves upon the
$O\big(m / \epsilon^2\big)$ bound given by previous work.

{\em Privelet} with the nominal wavelet transform runs in $O(n + m)$
time. In particular, computing $M$ from $T$ takes $O(n)$ time; the
nominal wavelet transform on $M$ has an $O(m)$ complexity. The noise
injection step incurs $O(m)$ overhead. Finally, with a breath-first
traversal of the decomposition tree $R$, we can complete both the mean
subtraction procedure and the reconstruction of the noisy frequency
matrix. Such a breath-first traversal takes $O(m)$ time under the
realistic assumption that the number of internal nodes in $R$ is
$O(m)$.

\subsection{Nominal Wavelet Transform vs.\ Haar Wavelet Transform} \label{sec:nom-compare}

As discussed in Section~\ref{sec:nom-wave}, {\em Privelet} with the
HWT can provide an $O\big((\log_2 m)^3 / \epsilon^2 \big)$ noise
variance bound for one-dimensional nominal data by imposing a total
order on the nominal domain. Asymptotically, this bound is inferior to
the $O(h^2 / \epsilon^2)$ bound in Equation~\ref{eqn:nom-bound}, but
how different are they in practice? To answer this question, let us
consider the nominal attribute {\em Occupation} in the Brazil census
dataset used in our experiments (see Section~\ref{sec:exp} for
details). It has a domain with $m=512$ leaves and a hierarchy with $3$
levels. Suppose that we apply {\em Privelet} with the one-dimensional
HWT on a dataset that contains {\em Occupation} as the only
attribute. Then, by Equation~\ref{eqn:ord-bound}, we can achieve a
noise variance bound of
\begin{eqnarray*}
(2 + \log_2 m) \cdot (2 + 2 \log_2 m)^2 / \epsilon^2 \; = \; 4400/\epsilon^2.
\end{eqnarray*}
In contrast, if we use {\em Privelet} with the nominal wavelet transform, the resulting noise variance is bounded by
\begin{eqnarray*}
4 \cdot 2 \cdot (2h)^2 / \epsilon^2 &=& 288/\epsilon^2,
\end{eqnarray*}
i.e., we can obtain a $15$-fold reduction in noise variance. Due to the superiority of the nominal wavelet transform over
the straightforward HWT, in the remainder of paper we will always use the former for nominal attributes.

\section{Multi-Dimensional Privelet} \label{sec:multi}

This section extends {\em Privelet} for multi-dimensional
data. Section~\ref{sec:multi-multi} presents our multi-dimensional
wavelet transform, which serves as the basis of the new instantiation
of {\em Privelet} in
Section~\ref{sec:multi-noise}. Section~\ref{sec:multi-bounds} analyzes
properties of the new instantiation, while
Section~\ref{sec:multi-hybrid} further improves its utility guarantee.

\subsection{Multi-Dimensional Wavelet Transform} \label{sec:multi-multi}

The one-dimensional wavelet transforms can be extended to
multi-dimensional data using {\em standard decomposition}
\cite{sds96}, which works by applying the one-dimensional wavelet
transforms along each dimension of the data in turn. More
specifically, given a frequency matrix $M$ with $d$ dimensions, we
first divide the entries in $M$ into one-dimensional vectors, such
that each vector contains a maximal set of entries that have identical
coordinates on all but the first dimensions. For each vector $V$, we
convert it into a set $S$ of wavelet coefficients using the
one-dimensional Haar or nominal wavelet transform, depending on
whether the first dimension of $M$ is ordinal or nominal. After that,
we store the coefficients in $S$ in a vector $V'$, where the
coefficients are sorted based on a level-order traversal of the
decomposition tree (the base coefficient always ranks first). The
$i$-th ($i \in [1, S]$) coefficient in $V'$ is assigned $d$
coordinates $\langle i, x_2, x_3, \ldots, x_d \rangle$, where $x_j$ is
the $j$-th coordinate of the entries in $V$ ($j \in [2, d]$; recall
that the $j$-th coordinates of these entries are identical). After
that, we organize all wavelet coefficients into a new $d$-dimensional
matrix $C_1$ according to their coordinates.

In the second step, we treat $C_1$ as the input data, and apply a
one-dimensional wavelet transform along the second dimension of $C_1$
to produce a matrix $C_2$, in a manner similar to the transformation
from $M$ to $C_1$. In general, the matrix $C_i$ generated in the
$i$-th step will be used as the input to the $(i+1)$-th step. In turn,
the $(i+1)$-th step will apply a one-dimensional wavelet transform
along the $(i+1)$-th dimension of $C_i$, and will generate a new
matrix $C_{i+1}$. We refer to $C_i$ as the {\em step-$i$
  matrix}. After all $d$ dimensions are processed, we stop and return
$C_d$ as the result. We refer to the transformation from $M$ to $C_d$
as an Haar-nominal (HN) wavelet transform. Observe that $C_d$ can be
easily converted back to the original matrix $M$, by applying inverse
wavelet transforms along dimensions $d, d-1, \ldots, 1$ in turn.

\begin{figure}[t]
\centering
\includegraphics[width=80mm]{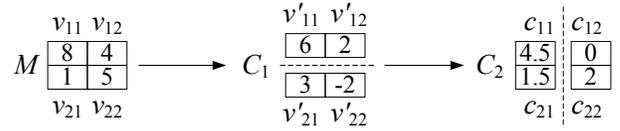}
\figcapup \caption{Multi-Dimensional Wavelet Transform} \figcapdown
\label{fig:multi-wave}
\end{figure}

\extraspacing
\begin{example}
Figure~\ref{fig:multi-wave} illustrates an HN wavelet transform on a
matrix $M$ with two ordinal dimensions. In the first step of the
transform, $M$ is vertically divided into two vectors $\langle v_{11},
v_{12} \rangle$ and $\langle v_{21}, v_{22} \rangle$. These two
vectors are then converted into two new vectors $\langle v'_{11},
v'_{12} \rangle$ and $\langle v'_{21}, v'_{22} \rangle$ using the
one-dimensional HWT. Note that $v'_{11}$ and $v'_{21}$ are the base
coefficients. The new matrix $C_1$ is the step-$1$ matrix.

Next, $C_1$ is horizontally partitioned into two vectors $\langle
v'_{11}, v'_{21} \rangle$ and $\langle v'_{12}, v'_{22} \rangle$. We
apply the HWT on them, and generate two coefficient vectors $\langle
c_{11}, c_{21} \rangle$ and $\langle c_{12}, c_{22} \rangle$, with
$c_{11}$ and $c_{12}$ being the base coefficients. The matrix $C_2$ is
returned as the final result. \done
\end{example}

\subsection{Instantiation of Privelet} \label{sec:multi-noise}

Given a $d$-dimensional table $T$ and a parameter $\lambda$, {\em
  Privelet} first performs the HN wavelet transform on the frequency
matrix $M$ of $T$. Then, it adds a Laplace noise with magnitude
$\lambda / \mathcal{W}_{HN}(c)$ to each coefficient $c$, where
$\mathcal{W}_{HN}$ is a weight function that we will define shortly. Next,
it reconstructs a noisy frequency matrix $M^*$ using the noisy wavelet
coefficients by inverting the one-dimensional wavelet transforms on
dimensions $d, d-1, \ldots, 1$ in turn.\footnote{If the $i$-the
  dimension is nominal, then, whenever we convert a vector $V'$ in the
  step-$i$ matrix back to a vector $V$ in the step-$(i-1)$ matrix, we
  will apply the mean substraction procedure before the reconstruction
  of $V$.} Finally, it terminates by returning $M^*$.

The weight function $\mathcal{W}_{HN}$ is decided by the
one-dimensional wavelet transforms adopted in the HN wavelet
transform. Let $\mathcal{W}_i$ be the weight function associated with
the transform used to compute the step-$i$ ($i \in [1, d]$) matrix,
i.e., $\mathcal{W}_i = \mathcal{W}_{Haar}$ if the $i$-th dimension of
$M$ is ordinal, otherwise $\mathcal{W}_i = \mathcal{W}_{Nom}$. We
determine the weight $\mathcal{W}_{HN}(c)$ for each HN wavelet
coefficient $c$ as follows.  First, during the construction of the
step-$1$ matrix $C_1$, whenever we generate a coefficient vector $V'$,
we assign to each $c' \in V'$ a weight $\mathcal{W}_1(c')$. For
instance, if the first dimension $A_1$ of $M$ is nominal, then
$\mathcal{W}_1(c') = 1$ if $c$ is the base coefficient, otherwise
$\mathcal{W}_1(c') = f/(2f - 2)$, where $f$ is the fanout of the
parent of $c'$ in the decomposition tree. Due to the way we arrange
the coefficients in $C_1$, if two coefficients in $C_1$ have the same
coordinates on the first dimension, they must have identical weights.

Now consider the second step of the HN wavelet transform. In this
step, we first partition $C_1$ into vectors along the second
dimension, and then apply one-dimensional wavelet transforms to convert
each vector $V''$ to into a new coefficient vector $V^*$. Observe that
all coefficients in $V''$ should have the same weight, since they have
identical coordinates on the first dimension. We set the weight of
each $c^* \in V^*$ to be $\mathcal{W}_2(c^*)$ times the weight shared
by the coefficients in $V''$.

In general, in the $i$-th step of the HN wavelet transform, whenever
we generate a coefficient $c$ from a vector $V \subset C_{i-1}$, we
always set the weight of $c$ to the product of $\mathcal{W}_i(c)$ and
the weight shared by the coefficients in $V$ --- all coefficients in
$V$ are guaranteed to have the same weight, because of the way we
arrange the entries in $C_{i-1}$. The weight function
$\mathcal{W}_{HN}$ for the HN wavelet transform is defined as a
function that maps each coefficient in $C_d$ to its weight computed as
above. For convenience, for each coefficient $c \in C_i$ ($i \in [1,
  d-1]$), we also use $\mathcal{W}_{HN}(c)$ to denote the weight of
$c$ in $C_i$.

\extraspacing
\begin{example}
Consider the HN wavelet transform in Figure~\ref{fig:multi-wave}. Both
dimensions of the frequency matrix $M$ are nominal, and hence, the
weight function for both dimensions is $\mathcal{W}_{Haar}$. In the
step-$1$ matrix $C_1$, the weights of the coefficients $v'_{11}$ and
$v'_{21}$ equal $1/2$, because (i) they are the base coefficients in
the wavelet transforms on $\langle v_{11}, v_{12} \rangle$ and
$\langle v_{21}, v_{21}\rangle$, respectively, and (ii)
$\mathcal{W}_{Haar}$ assigns a weight $1/2$ to the base coefficient
whenever the input vector contains only two entries.

Now consider the coefficient $c_{11}$ in the step-$2$ matrix $C_2$. It
is generated from the HWT on $\langle v'_{11},
v'_{21} \rangle$, where both $v'_{11}$ and $v'_{12}$ have a weight
$1/2$. In addition, as $c_{11}$ is the base coefficient,
$\mathcal{W}_{Haar}(c_{11}) = 1/2$. Consequently,
$\mathcal{W}_{HN}(c_{11}) = 1/2 \cdot \mathcal{W}_{Haar}(c_{11}) =
1/4$. \done
\end{example}
%\extraspacing

\subsection{Theoretical Analysis} \label{sec:multi-bounds}

As {\em Privelet} with the HN wavelet transform is essentially a
composition of the solutions in Sections \ref{sec:ord-prive} and
\ref{sec:nom-noise}, we can prove its privacy (utility) guarantee by
incorporating Lemmas \ref{lmm:ord-amp} and \ref{lmm:nom-amp}
(\ref{lmm:ord-utility} and \ref{lmm:nom-utility}) with an induction
argument on the dataset dimensionality $d$. Let us define a function
$\mathcal{P}$ that takes as input any attribute $A$, such that
\begin{equation*}
\mathcal{P}(A) = \left\{
\begin{array}{ll}
1 + \log_2 |A| & \textrm{if $A$ is ordinal} \\
\textrm{the height $h$ of $A$'s hierarchy } & \textrm{otherwise}
\end{array}
\right.
\end{equation*}
Similarly, let $\mathcal{H}$ be a function such that
\begin{equation*}
\mathcal{H}(A) = \left\{
\begin{array}{ll}
(2 + \log_2 |A|)/2 & \textrm{if $A$ is ordinal} \\
4 & \textrm{otherwise}
\end{array}
\right.
\end{equation*}
We have the following theorems that show (i) the generalized sensitivity of the HN wavelet transform (Theorem
\ref{thrm:multi-amp}) and (ii) the noise variance bound provided by {\em Privelet} with the HN wavelet transform
(Theorem \ref{thrm:multi-utility}).

\extraspacing
\begin{theorem} \label{thrm:multi-amp}
The HN wavelet transform on a $d$-dimensional matrix $M$ has a
generalized sensitivity $\prod_{i=1}^d \mathcal{P}(A_i)$ with respect
to $\mathcal{W}_{HN}$, where $A_i$ is the $i$-th dimension of
$M$.
\end{theorem}
\begin{proof}
See Appendix~\ref{sec:proof-multi-amp}.
\end{proof}

\extraspacing
\begin{theorem} \label{thrm:multi-utility}
Let $C^*_d$ be a $d$-dimensional HN wavelet coefficient matrix, such
that each coefficient $c^* \in C^*_d$ has a noise with a variance at
most $\big(\sigma / \mathcal{W}_{HN}(c^*)\big)^2$. Let $M^*$ be the
noisy frequency matrix reconstructed from $C^*_d$, and $A_i$ ($i \in
[1, d]$) be the $i$-th dimension of $M^*$. For any range-count query
answered using $M^*$, the noise in the query result has a variance at
most $\sigma^2 \cdot \prod_{i=1}^d \mathcal{H}(A_i)$.
\end{theorem}
\begin{proof}
See Appendix~\ref{sec:proof-multi-utility}.
\end{proof}
\extraspacing

By Theorem~\ref{thrm:multi-amp}, to achieve $\epsilon$-differential
privacy, {\em Privelet} with the HN wavelet transform should be
applied with $\lambda = 2 / \epsilon \cdot \prod_{i=1}^d
\mathcal{P}(A_i)$; in that case, by Theorem~\ref{thrm:multi-utility},
        {\em Privelet} ensures that any range-count query result has a
        noise variance of at most
\begin{eqnarray*} %\label{eqn:multi-bound}
2 \bigg(2/\epsilon \cdot \prod_{i=1}^d \mathcal{P}(A_i)\bigg)^2 \cdot \prod_{i=1}^d \mathcal{H}(A_i) \; = \; O\left(\log^{O(1)} m / \epsilon^2 \right),\!
\end{eqnarray*}
since $\mathcal{P}(A_i)$ and $\mathcal{H}(A_i)$ are logarithmic in $m$.

{\em Privelet} with the HN wavelet transform has an $O(n + m)$ time
complexity. This is because (i) computing the frequency matrix $M$
takes $O(n + m)$ time, (ii) each one-dimensional wavelet transform on
$M$ has $O(m)$ complexity, and (iii) adding Laplace noise to the
wavelet coefficients incurs $O(m)$ overhead.

\subsection{A Hybrid Solution} \label{sec:multi-hybrid}

We have shown that {\em Privelet} outperforms Dwork et al.'s method
{\em asymptotically} in terms of the accuracy of range-count
queries. In practice, however, {\em Privelet} can be inferior to Dwork
et al.'s method, when the input table $T$ contains attributes with
small domains. For instance, if $T$ has a single ordinal attribute $A$
with domain size $|A| = 16$, then {\em Privelet} provides a noise
variance bound of
\begin{eqnarray*}
2 \cdot \big(2 \cdot \mathcal{P}(A) / \epsilon \big)^2 \cdot \mathcal{H}(A) &=& 600 / \epsilon^2,
\end{eqnarray*}
as analyzed in Section~\ref{sec:multi-bounds}. In contrast, Dwork et
al.'s method incurs a noise variance of at most
\begin{eqnarray*}
2 \cdot \big(2 \cdot |A| / \epsilon )^2  &=& 128 / \epsilon^2,
\end{eqnarray*}
as shown in Section~\ref{sec:prelim-basic}. This demonstrates the fact
that, Dwork et al.'s method is more favorable for small-domain
attributes, while {\em Privelet} is more suitable for attributes whose
domains are large. How can we combine the advantages of both solutions
to handle datasets that contain both large- and small-domain
attributes?

\begin{figure} [t]
\centering
\begin{small}
\begin{tabbing}
\hspace{4mm} \= \hspace{2mm} \= \hspace{2mm} \= \hspace{2mm} \=
\kill
{\bf Algorithm} {\em Privelet$^+$} ($T$, $\lambda$, $S_A$)  \\
1.\> map $T$ to its frequency matrix $M$ \\
2.\> divide $M$ into sub-matrices along the dimensions specified in $S_A$ \\
3.\> for each sub-matrix \\
4.\>\> compute the HN wavelet coefficients of the sub-matrix \\
5.\>\> add to each coefficient $c$ a Laplace noise with magnitude \\
\>\>  $\lambda / \mathcal{W}_{HN}(c)$ \\
6.\>\> convert the noisy coefficients back to a noisy sub-matrix \\
7.\> assemble the noisy sub-matrices into a frequency matrix $M^*$ \\
8.\> return $M^*$
\end{tabbing}
\end{small}
\codecapup  \caption{The {\em Privelet$^+$} algorithm} \codecapdown
\label{code:multi-hybrid}
\end{figure}

We answer the above question with the {\em Privelet$^+$} algorithm illustrated in Figure~\ref{code:multi-hybrid}. The algorithm takes as an input a table $T$, a parameter $\lambda$, and a subset $S_A$ of the attributes in $T$. It first maps $T$ to its frequency matrix $M$. Then, it divides $M$ into sub-matrices, such that each sub-matrix contains the entries in $M$ that have the same coordinates on each dimension specified in $S_A$. For instance, given the frequency matrix in Table~\ref{tbl:intro-freq}, if $S_A$ contains only the ``Has Diabetes?'' dimension, then the matrix would be split into two one-dimensional sub-matrices, each of which contains a column in Table~\ref{tbl:intro-freq}. In general, if $M$ has $d$ dimensions, then each sub-matrix should have $d - |S_A|$ dimensions.

After that, each sub-matrix is converted into wavelet coefficients
using a $(d-|S_A|)$-dimensional HN wavelet transform. {\em
  Privelet$^+$} injects into each coefficient $c$ a Laplace noise with
magnitude $\lambda / \mathcal{W}_{HN}(c)$, and then maps the noisy
coefficients back to a noisy sub-matrix. In other words, {\em
  Privelet$^+$} processes each sub-matrix in the same way as {\em
  Privelet} handles a $(d-|S_A|)$-dimensional frequency
matrix. Finally, {\em Privelet$^+$} puts together all noisy
sub-matrices to obtain a $d$-dimensional noisy frequency matrix $M^*$,
and then terminates by returning $M^*$.

Observe that {\em Privelet$^+$} captures {\em Privelet} as a special
case where $S_A = \emptyset$. Compared to {\em Privelet}, it provides
the flexibility of not applying wavelet transforms on the attributes
in $S_A$. Intuitively, this enables us to achieve better data utility
by putting in $S_A$ the attributes with small domains, since those
attributes cannot be handled well with {\em Privelet}. Our intuition
is formalized in Corollary~\ref{coro:multi-hybrid}, which follows from
Theorems \ref{thrm:multi-amp} and \ref{thrm:multi-utility}.

\extraspacing
\begin{corollary} \label{coro:multi-hybrid}
Let $T$ be a table that contains a set $S$ of attributes. Given $T$, a subset $S_A$ of $S$, and a parameter $\lambda$, {\em Privelet$^+$} achieves $\epsilon$-differential privacy with $\epsilon = 2/\lambda \cdot \prod_{A \in S - S_A} \mathcal{P}(A)$. In addition, it ensures that any range-count query result has a noise variance at most $\left(\prod_{A \in S_A} |A| \right) \cdot \prod_{A \in S - S_A} \mathcal{H}(A)$. \done
\end{corollary}
\extraspacing

By Corollary~\ref{coro:multi-hybrid}, when $\epsilon$-differential
privacy is enforced, {\em Privelet$^+$} leads to a noise variance
bound of
\begin{equation}  \label{eqn:multi-hybird-bound}
8/ \epsilon^2 \cdot \Big(\prod_{A \in S_A} |A| \Big) \cdot \prod_{A \in S - S_A} \Big( \left( \mathcal{P}(A) \right)^2 \cdot \mathcal{H}(A) \Big).
\end{equation}
It is not hard to verify that, when $S_A$ contains only attributes $A$
with $|A| \le \left( \mathcal{P}(A) \right)^2 \cdot \mathcal{H}(A)$,
the bound given in Equation~\ref{eqn:multi-hybird-bound} is always no
worse than the noise variance bounds provided by {\em Privelet} and
Dwork et al.'s method.

Finally, we note that {\em Privelet$^+$} also runs in $O(n + m)$
time. This follows from the $O(n + m)$ time complexity of {\em
  Privelet}.

\section{Experiments} \label{sec:exp}

This section experimentally evaluates {\em Privelet$^+$} and Dwork et
al.'s method (referred as {\em Basic} in the
following). Section~\ref{sec:exp-query} compares their data utility,
while Section~\ref{sec:exp-time} investigates their computational cost.

\subsection{Accuracy of Range-Count Queries} \label{sec:exp-query}

We use two datasets\footnote{Both datasets are public available as
  part of the {\em Integrated Public Use Microdata Series} \cite{mpc09}.} that contain census records of
individuals from Brazil and the US, respectively. The Brazil dataset
has $10$ million tuples and four attributes, namely, {\em Age}, {\em
  Gender}, {\em Occupation}, and {\em Income}. The attributes {\em
  Age} and {\em Income} are ordinal, while {\em Gender} and {\em
  Occupation} are nominal. The US dataset also contains these four
attributes (but with slightly different domains), and it has $8$
million tuples. Table~\ref{tbl:exp-attribute} shows the domain sizes
of the attributes in the datasets. The numbers enclosed in parentheses
indicate the heights of the hierarchies associated with the nominal attributes.

\begin{table}[t]
\centering
\tblcapup \caption{Sizes of Attribute Domains} \tblcapdown \label{tbl:exp-attribute}
\begin{small}
\begin{tabular}{|c|c|c|c|c|}
\cline{2-5}
\multicolumn{1}{c|}{} & {\bf Age} & {\bf Gender} & {\bf Occupation} & {\bf Income} \\ \hline
{\bf Brazil} & $101$ & $2$ ($2$) & $512$ ($3$) & $1001$ \\ \hline
{\bf US} & $96$ & $2$ ($2$) & $511$ ($3$) & $1020$ \\ \hline
\end{tabular}
\end{small}
\tbldown
\end{table}

For each dataset, we create a set of $40000$ random range-count queries, such
that the number of predicates in each query is uniformly distributed
in $[1, 4]$. Each query predicate ``$A_i \in S_i$'' is generated as
follows. First, we choose $A_i$ randomly from the attributes in the
dataset. After that, if $A_i$ is ordinal, then $S_i$ is set to a
random interval defined on $A_i$; otherwise, we randomly select a
non-root node from the hierarchy of $A_i$, and let $S_i$ contain all
leaves in the subtree of the node. We define the {\em selectivity} of
a query $q$ as the fraction of tuples in the dataset that satisfy all
predicates in $q$. We also define the {\em coverage} of $q$ as the
fraction of entries in the frequency matrix that are covered by $q$.

We apply {\em Basic} and {\em Privelet$^+$} on each dataset to produce
noisy frequency matrices that ensure $\epsilon$-differential privacy,
varying $\epsilon$ from $0.5$ to $1.25$. For {\em Privelet$^+$}, we
set its input parameter ${S_A} = \{\textrm{\em Age}, \textrm{\em
  Gender}\}$, since each $A$ of these two attributes has a relatively
small domain, i.e., $|A| \le (\mathcal{P}(A))^2 \cdot \mathcal{H}(A)$,
where $\mathcal{P}$ and $\mathcal{H}$ are as defined in
Section~\ref{sec:multi-bounds}. We use the noisy frequency matrices to
derive approximate answers for range-count queries. The quality of
each approximate answer $x$ is gauged by its {\em square error} and
{\em relative error} with respect to the actual query result
$act$. Specifically, the square error of $x$ is defined as $(x -
act)^2$, and the relative error of $x$ is computed as $|x-act| /
\max\{act, s\}$, where $s$ is a {\em sanity bound} that
mitigates the effects of the queries with excessively small
selectivities (we follow with this evaluation methodology from previous
work \cite{vw99,gk05}). We set $s$ to $0.1\%$ of the number of tuples in the dataset.

In our first set of experiments, we divide the query set $Q_{Br}$ for
the Brazil dataset into $5$ subsets. All queries in the $i$-th ($i \in
[1, 5]$) subset have coverage that falls between the $(i-1)$-th and $i$-th
quintiles of the query coverage distribution in $Q_{Br}$. On each
noisy frequency matrix generated from the Brazil dataset, we process
the $5$ query subsets in turn, and plot in
Figure~\ref{fig:exp-rel-vs-sel-bra} the average square error in each
subset as a function of the average query
coverage. Figure~\ref{fig:exp-rel-vs-sel-usa} shows the results of a
similar set of experiments conducted on the US dataset.

\begin{figure*}[t]
\begin{small}
\begin{tabular}{cccc}
\hspace{-6mm}\includegraphics[height=35mm]{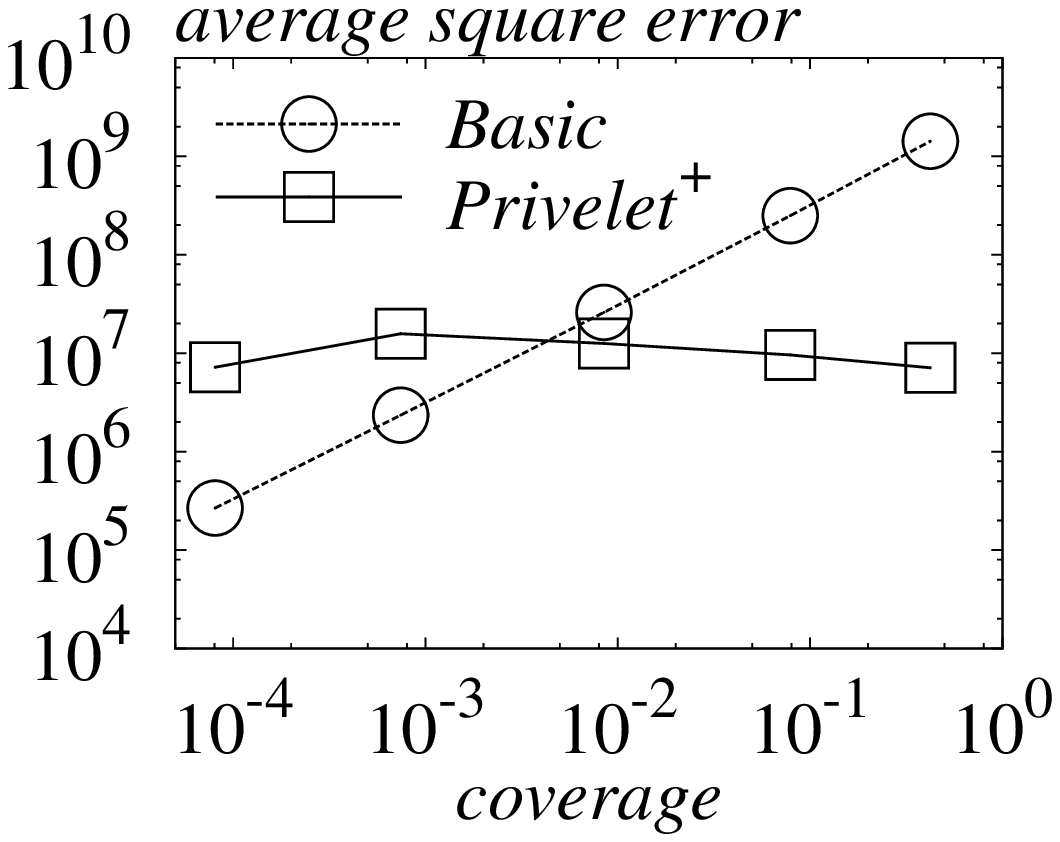} &
\hspace{-8mm}\includegraphics[height=35mm]{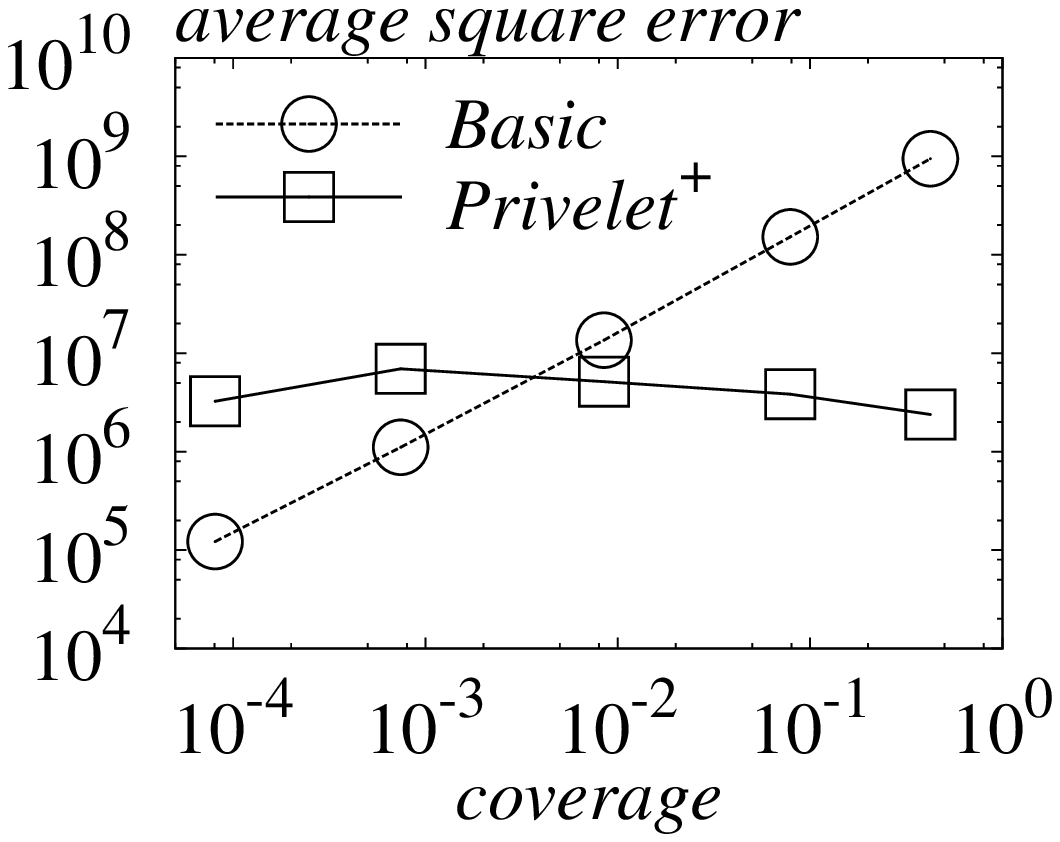} &
\hspace{-8mm}\includegraphics[height=35mm]{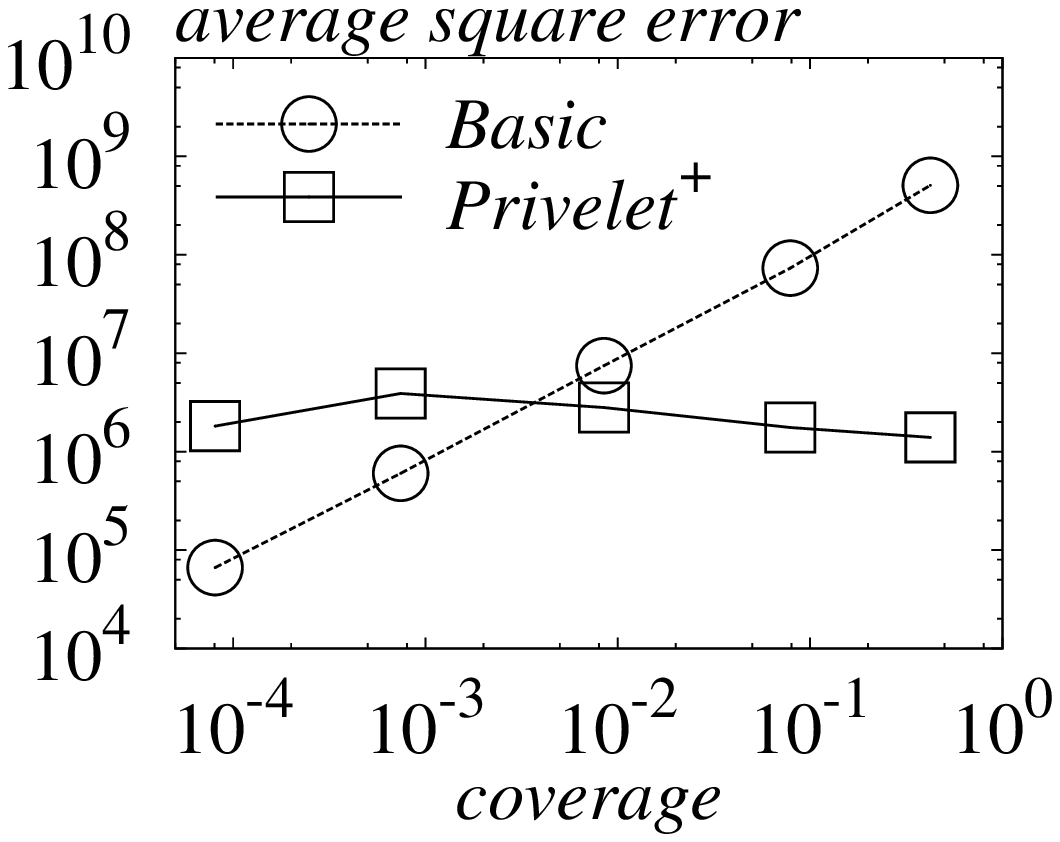} &
\hspace{-8mm}\includegraphics[height=35mm]{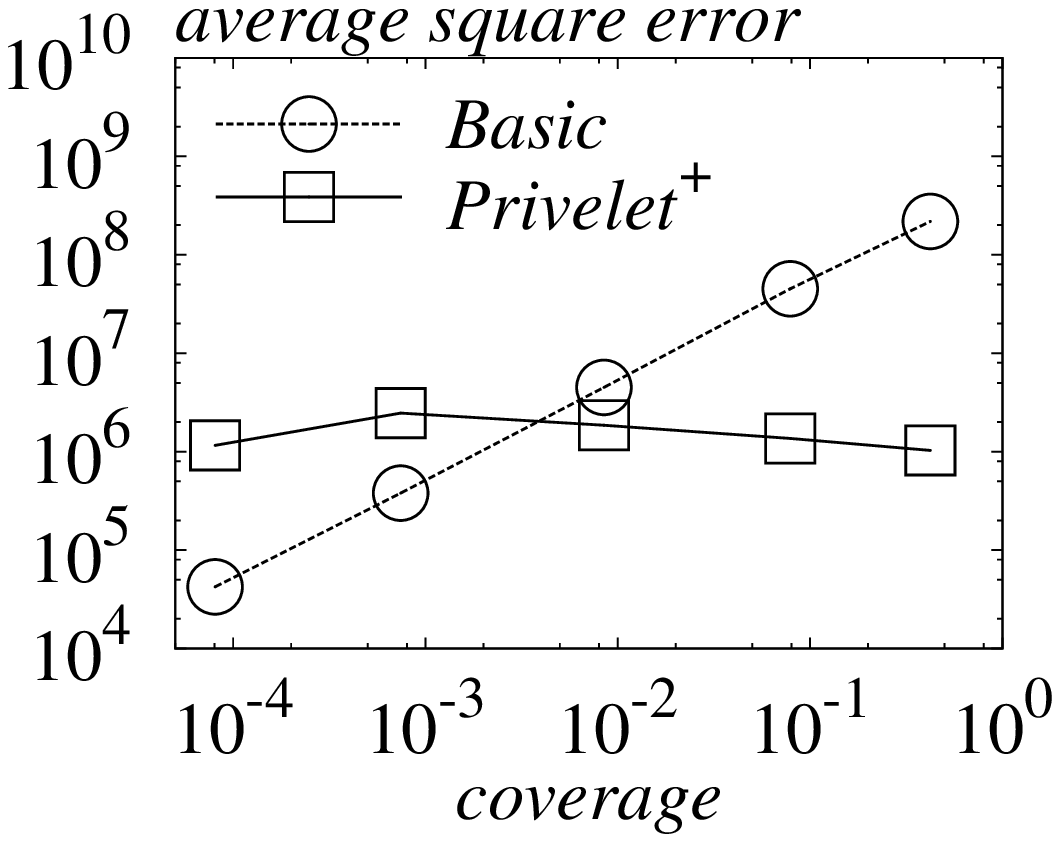} \\
\hspace{3mm}(a) $\epsilon = 0.5$ & \hspace{-0mm}(b) $\epsilon = 0.75$ & \hspace{-0mm}(c) $\epsilon = 1$ & \hspace{-0mm}(d) $\epsilon = 1.25$
\end{tabular}
\end{small}
\figcapup \caption{Average Square Error vs.\ Query Coverage (Brazil)} \figcapdown \vspace{2mm}
\label{fig:exp-var-vs-cov-bra}
\end{figure*}

\begin{figure*}[t]
\begin{small}
\begin{tabular}{cccc}
\hspace{-6mm}\includegraphics[height=35mm]{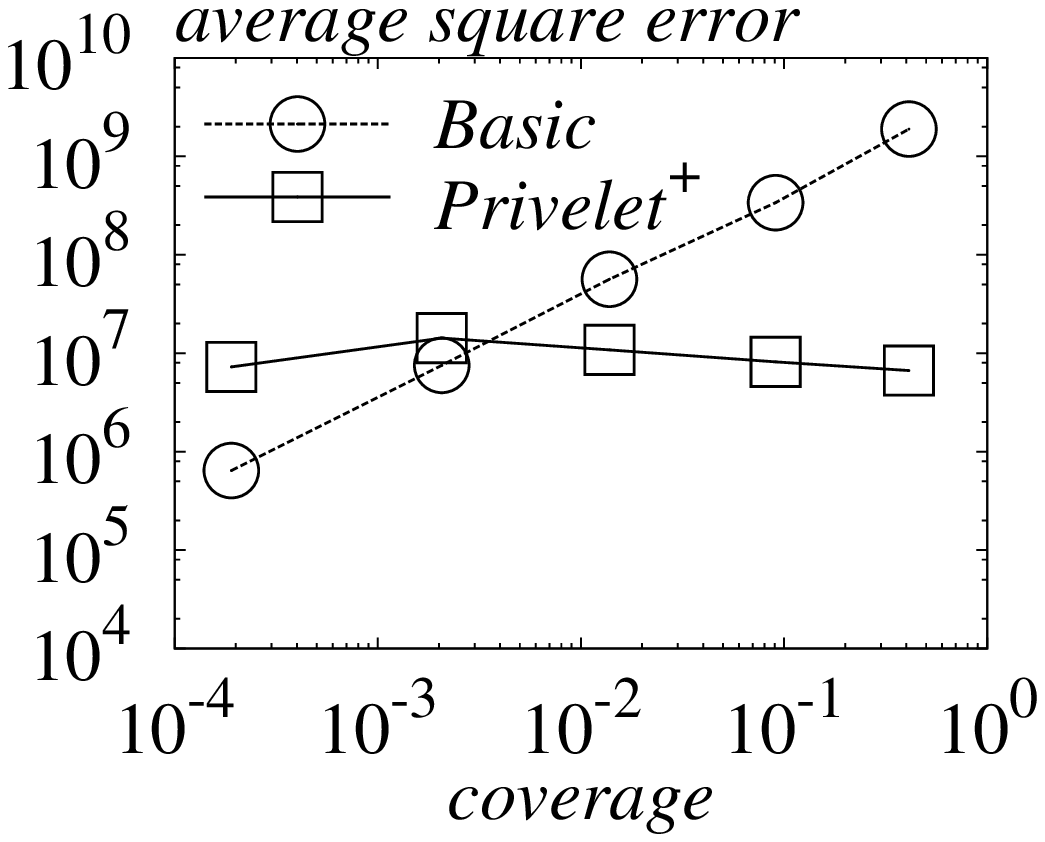} &
\hspace{-8mm}\includegraphics[height=35mm]{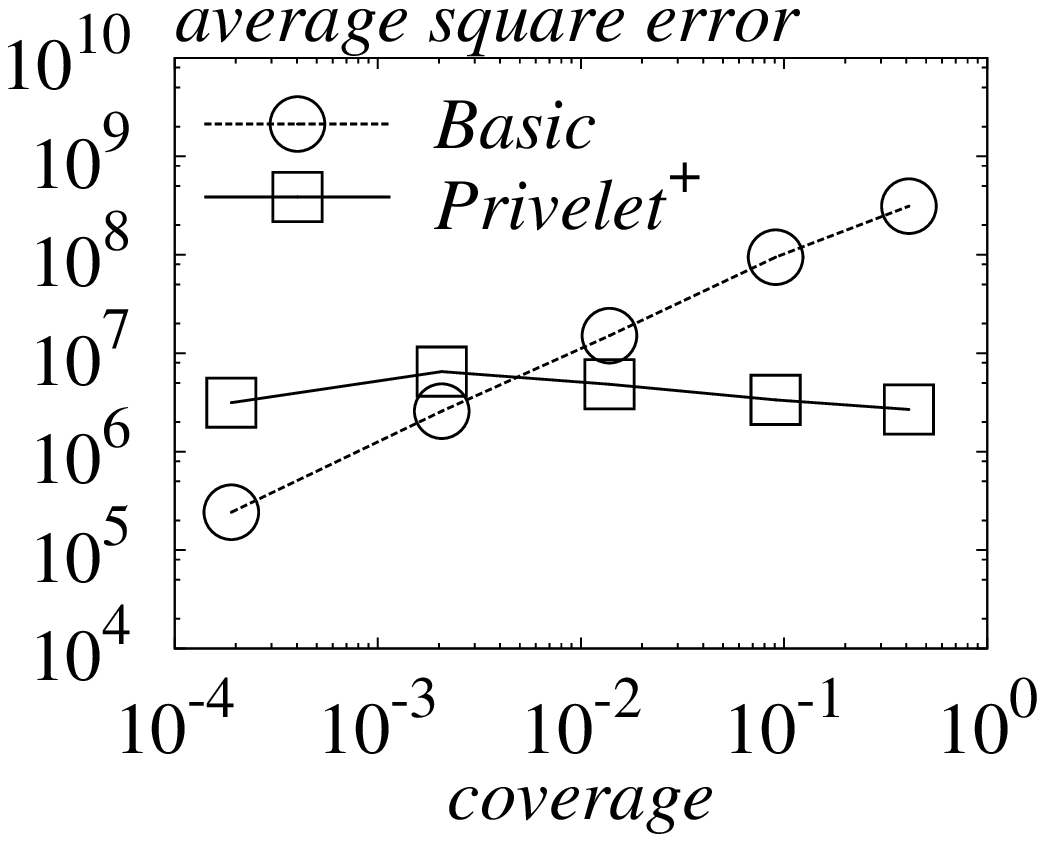} &
\hspace{-8mm}\includegraphics[height=35mm]{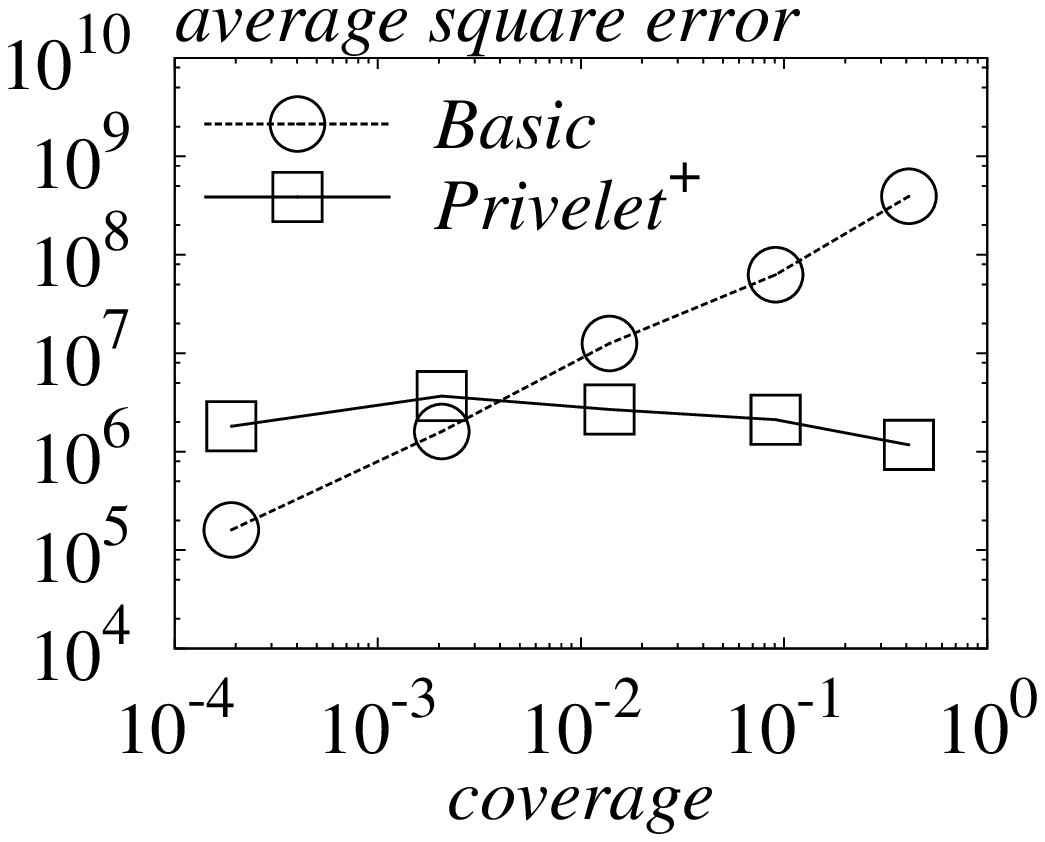} &
\hspace{-8mm}\includegraphics[height=35mm]{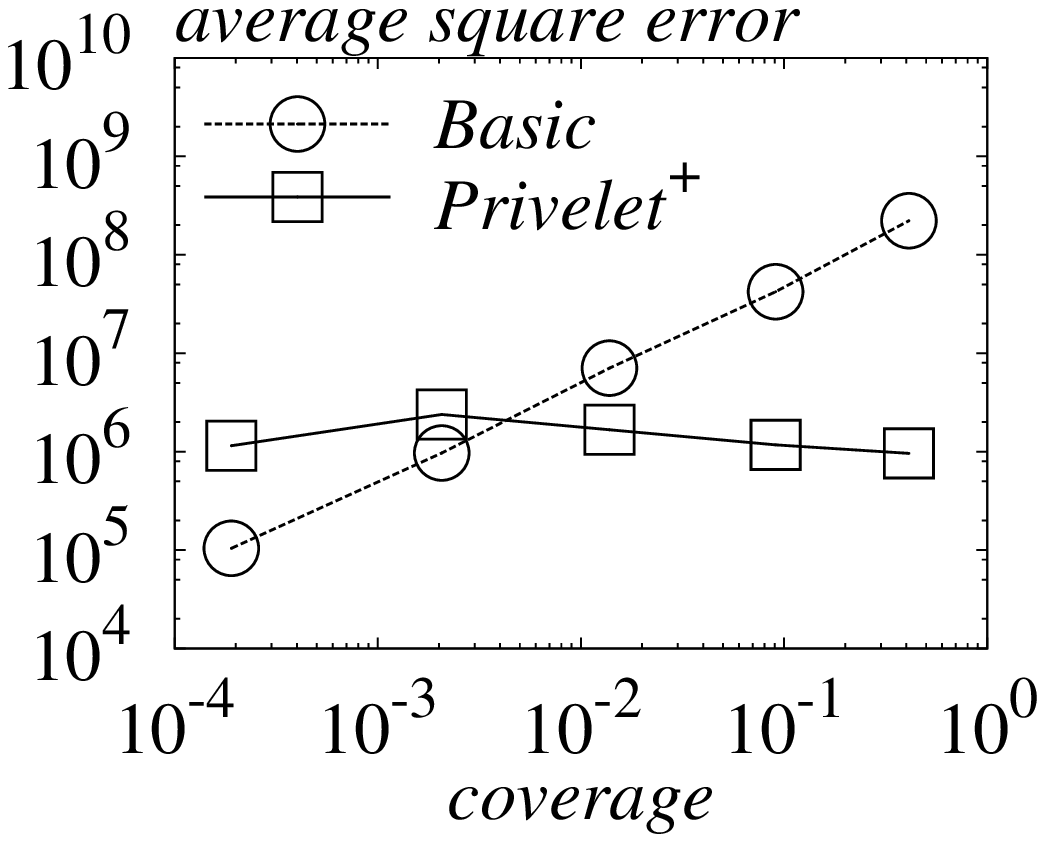} \\
\hspace{3mm}(a) $\epsilon = 0.5$ & \hspace{-0mm}(b) $\epsilon = 0.75$ & \hspace{-0mm}(c) $\epsilon = 1$ & \hspace{-0mm}(d) $\epsilon = 1.25$
\end{tabular}
\end{small}
\figcapup \caption{Average Square Error vs.\ Query Coverage (US)} \figcapdown \vspace{2mm}
\label{fig:exp-var-vs-cov-usa}
\end{figure*}

The average square error of {\em Basic} increases linearly with the
query coverage, which conforms to the analysis in
Section~\ref{sec:prelim-basic} that {\em Basic} incurs a noise
variance linear to the coverage of the query. In contrast, the average
square error of {\em Privelet$^+$} is insensitive to the query
coverage. The maximum average error of {\em Privelet$^+$} is smaller
than that of {\em Basic} by two orders of magnitudes. This is
consistent with our results that {\em Privelet$^+$} provides a much
better noise variance bound than {\em Basic} does. The error of both
{\em Basic} and {\em Privelet$^+$} decreases with the increase of
$\epsilon$, since a larger $\epsilon$ leads to a smaller amount of
noise in the frequency matrices.

\begin{figure*}[t]
\begin{small}
\begin{tabular}{cccc}
\hspace{-6mm}\includegraphics[height=35mm]{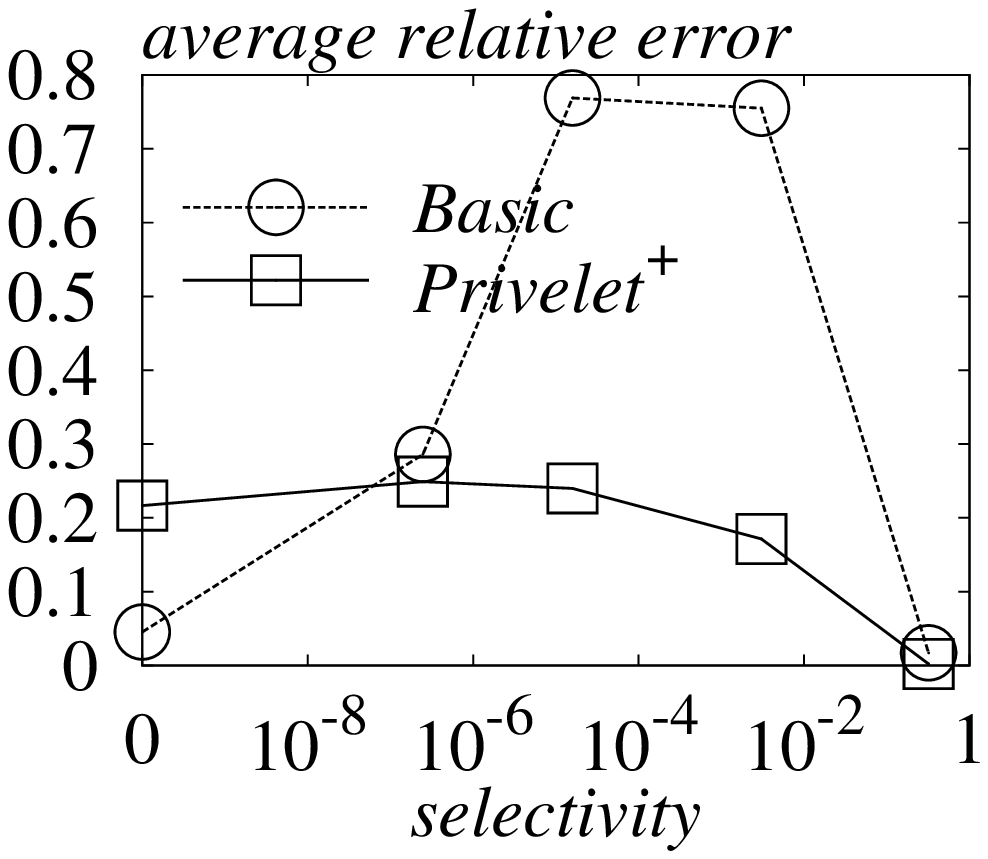} &
\hspace{-8mm}\includegraphics[height=35mm]{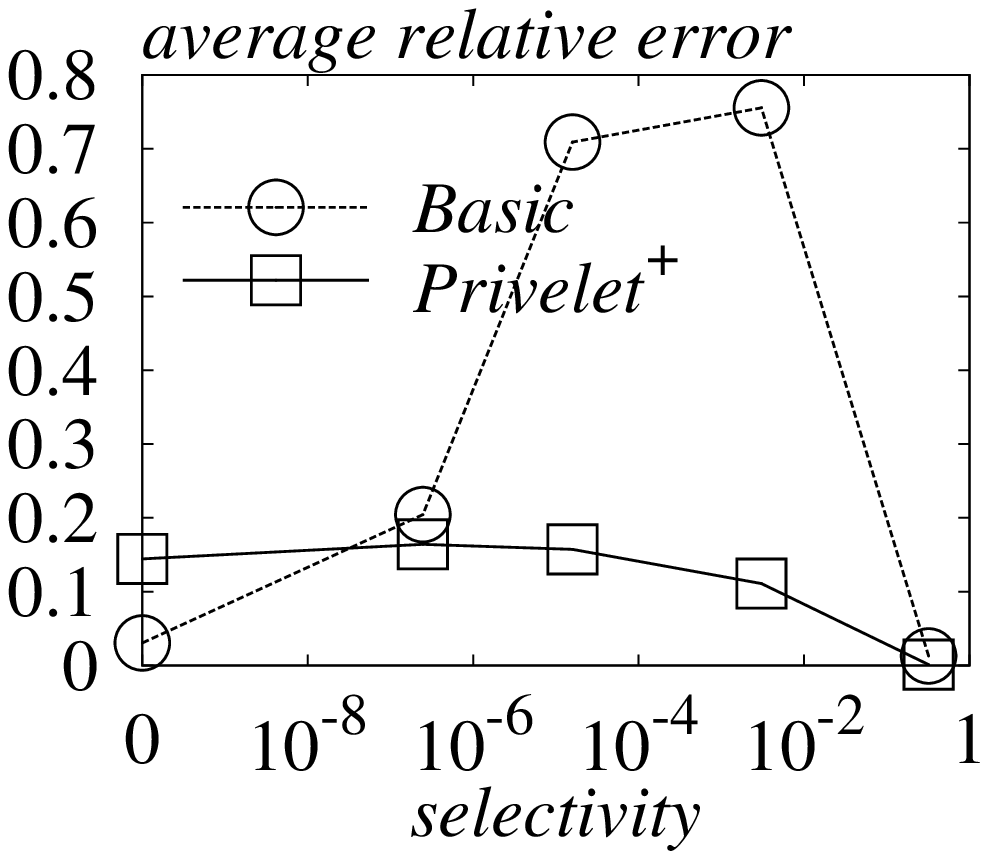} &
\hspace{-8mm}\includegraphics[height=35mm]{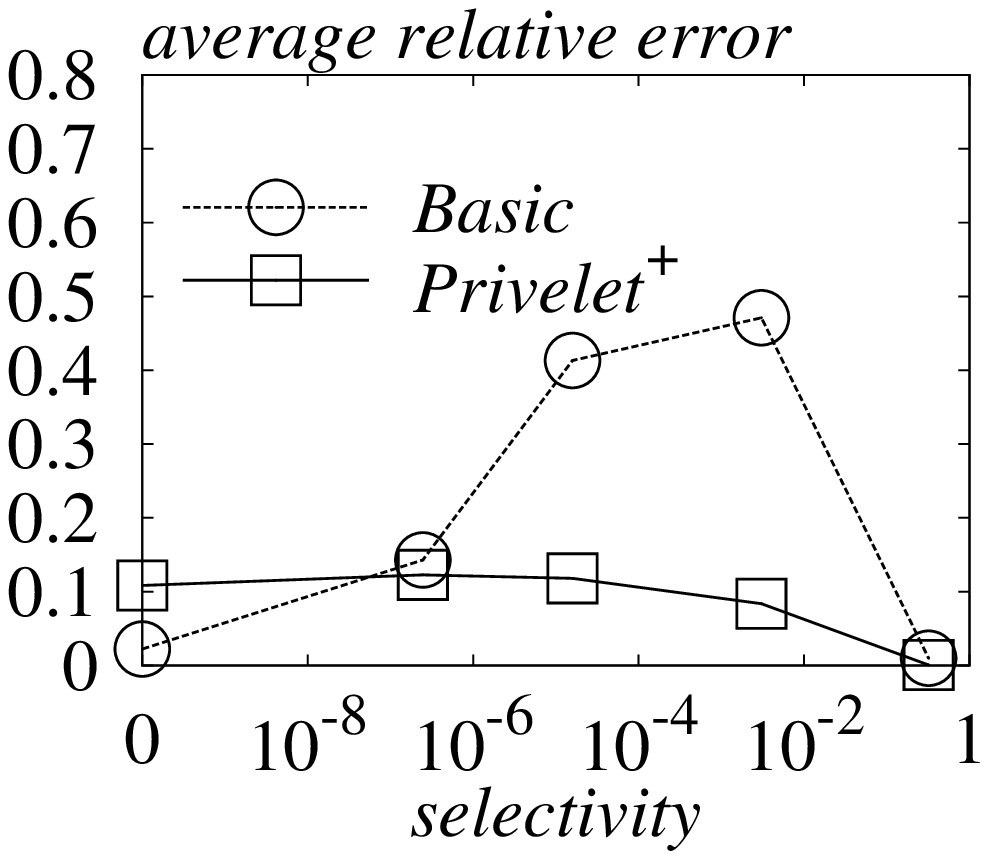} &
\hspace{-8mm}\includegraphics[height=35mm]{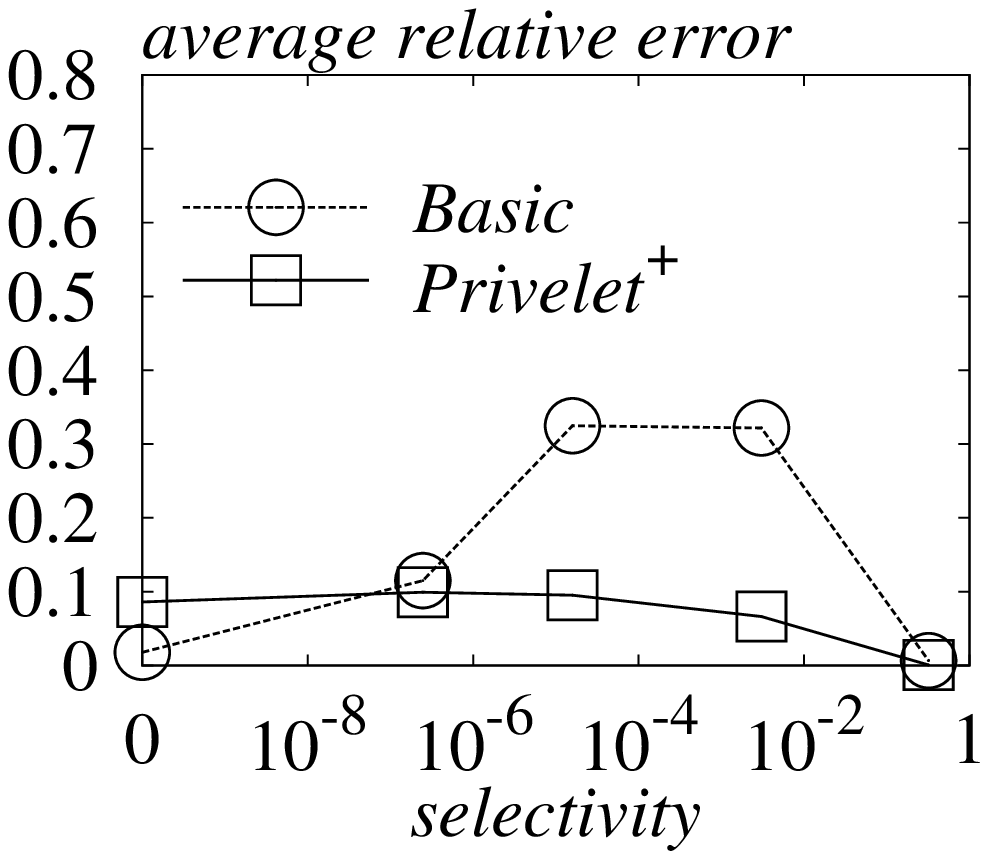} \\
\hspace{3mm}(a) $\epsilon = 0.5$ & \hspace{-0mm}(b) $\epsilon = 0.75$ & \hspace{-0mm}(c) $\epsilon = 1$ & \hspace{-0mm}(d) $\epsilon = 1.25$
\end{tabular}
\end{small}
\figcapup \caption{Average Relative Error vs.\ Query Selectivity (Brazil)} \figcapdown \vspace{2mm}
\label{fig:exp-rel-vs-sel-bra}
\end{figure*}

\begin{figure*}[t]
\begin{small}
\begin{tabular}{cccc}
\hspace{-6mm}\includegraphics[height=35mm]{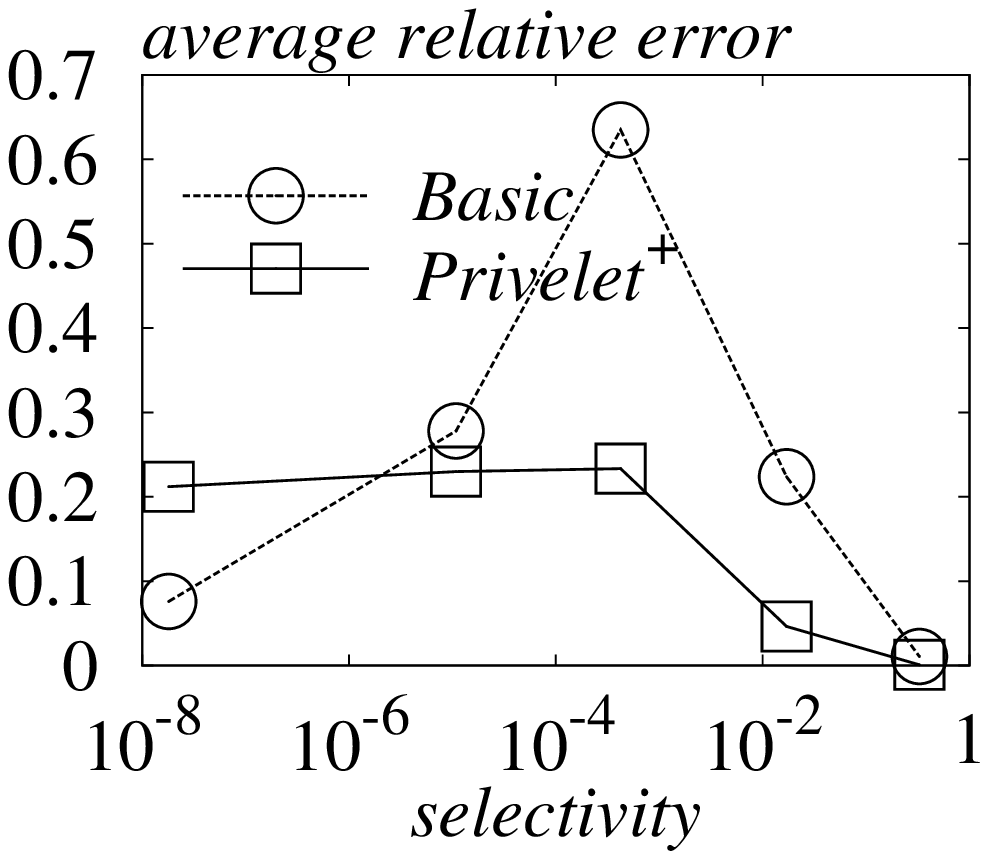} &
\hspace{-8mm}\includegraphics[height=35mm]{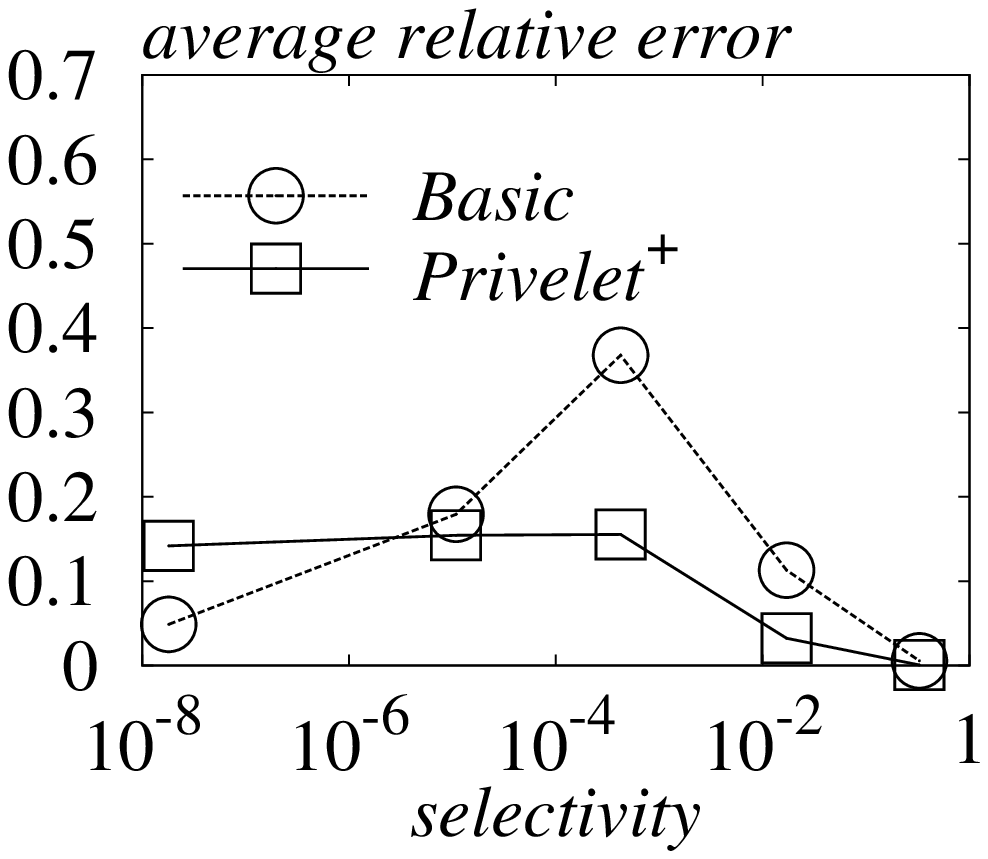} &
\hspace{-8mm}\includegraphics[height=35mm]{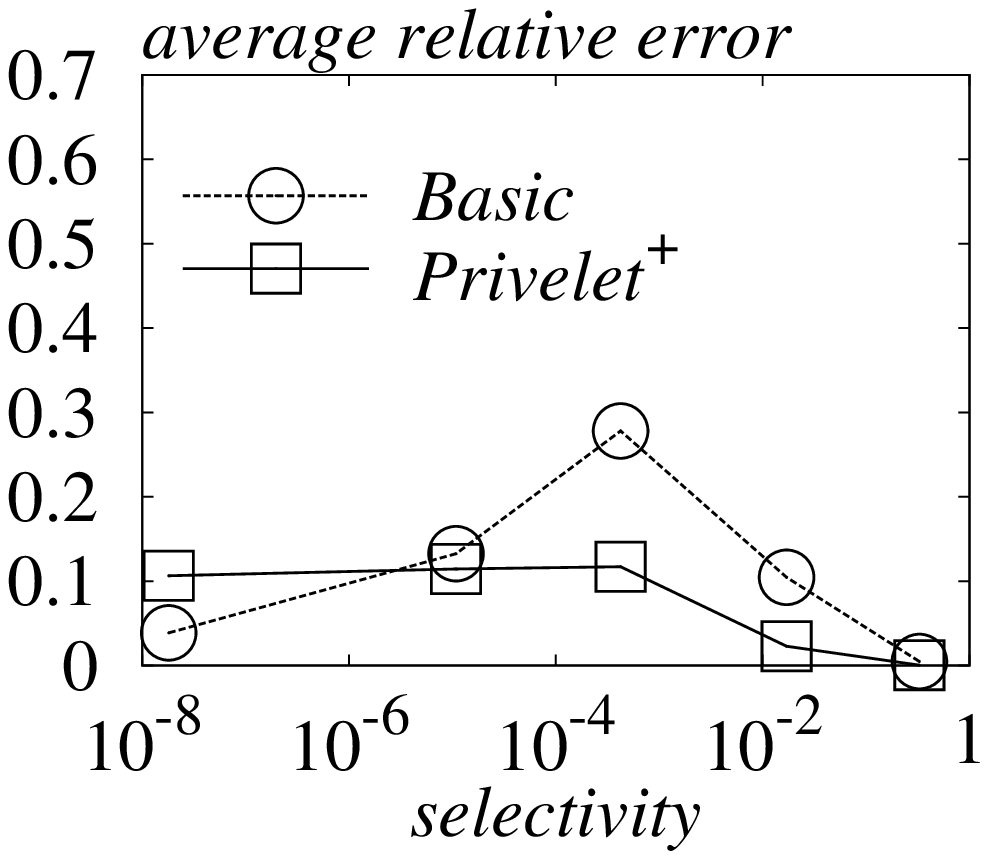} &
\hspace{-8mm}\includegraphics[height=35mm]{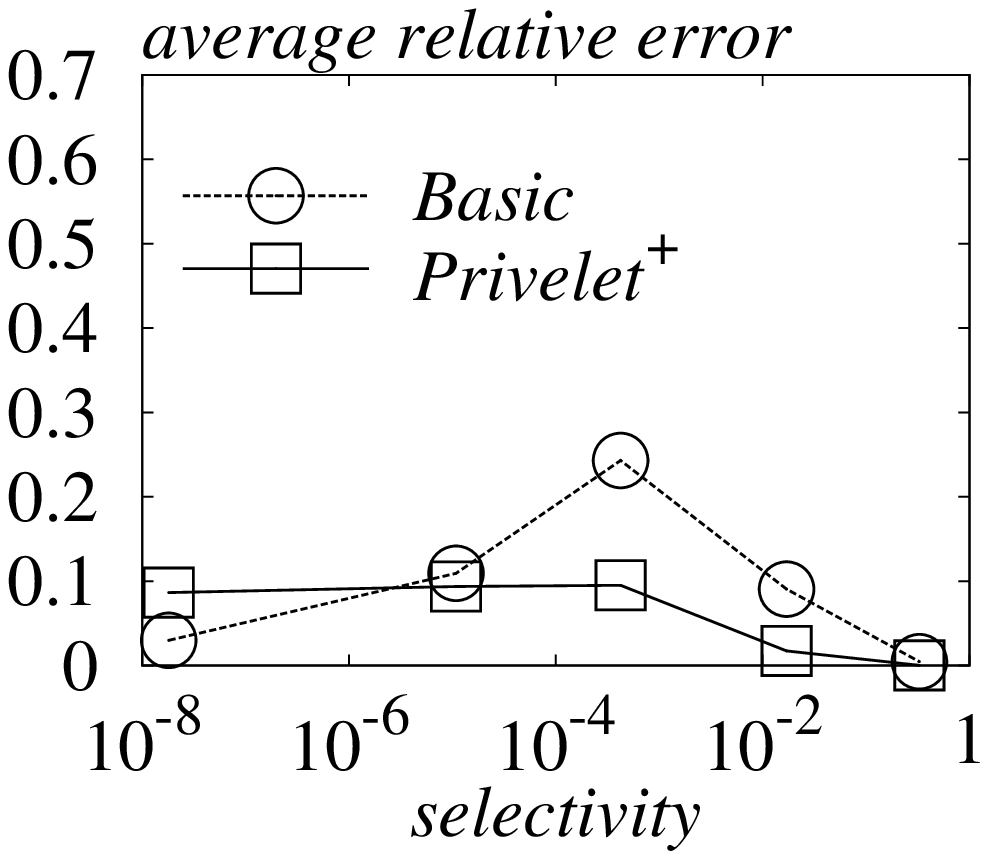} \\
\hspace{3mm}(a) $\epsilon = 0.5$ & \hspace{-0mm}(b) $\epsilon = 0.75$ & \hspace{-0mm}(c) $\epsilon = 1$ & \hspace{-0mm}(d) $\epsilon = 1.25$
\end{tabular}
\end{small}
\figcapup \caption{Average Relative Error vs.\ Query Selectivity (US)} \figcapdown
\label{fig:exp-rel-vs-sel-usa}
\end{figure*}

In the next experiments, we divide the query set for each dataset into
$5$ subsets based on query selectivities. Specifically, the $i$-th ($i
\in [1, 5]$) subset contains the queries whose selectivities are
between the $(i-1)$-th and $i$-th quintiles of the overall query
selectivity distribution. Figures \ref{fig:exp-rel-vs-sel-bra} and
\ref{fig:exp-rel-vs-sel-usa} illustrate the average relative error
incurred by each noisy frequency matrix in answering each query
subset. The X-axes of the figures represent the average selectivity of
each subset of queries. The error of {\em Privelet$^+$} is
consistently lower than that of {\em Basic}, except when the query
selectivity is below $10^{-7}$. In addition, the error of {\em
  Privelet$^+$} is no more than $25\%$ in all cases, while {\em Basic}
induces more than $70\%$ error in several query subsets.

In summary, our experiments show that {\em Privelet$^+$} significantly
outperforms {\em Basic} in terms of the accuracy of range-count
queries. Specifically, {\em Privelet$^+$} incurs a smaller query error
than {\em Basic} does, whenever the query coverage is larger than
$1\%$ or the query selectivity is at least $10^{-7}$.

\subsection{Computation Time} \label{sec:exp-time}

Next, we investigate how the computation time of {\em Basic} and {\em
  Privelet$^+$} varies with the number $n$ tuples in the input data
and the number $m$ of entries in the frequency matrix. For this
purpose, we generate synthetic datasets with various values of $n$ and
$m$. Each dataset contains two ordinal attributes and two nominal
attributes. The domain size of each attribute is $m^{1/4}$. Each
nominal attribute $A$ has a hierarchy $H$ with three levels, such that
the number of level-2 nodes in $H$ is $\sqrt{|A|}$. The values of the
tuples are uniformly distributed in the attribute domains.

In the first set of experiments, we fix $m = 2^{24}$, and apply {\em
  Basic} and {\em Privelet$^+$} on datasets with $n$ ranging from $1$
million to $5$ millions. For {\em Privelet$^+$}, we set its input
parameter $S_A = \emptyset$, in which case {\em Privelet$^+$} has a
relatively large running time, since it needs to perform wavelet
transforms on all dimensions of the frequency
matrix. Figure~\ref{fig:time-vs-n} illustrates the computation time of
{\em Basic} and {\em Privelet$^+$} as a function of $n$. Observe that
both techniques scale linearly with $n$.

In the second set of experiments, we set $n = 5 \times 10^6$, and vary
$m$ from $2^{22}$ to $2^{26}$. Figure~\ref{fig:time-vs-m} shows the
computation overhead of {\em Basic} and {\em Privelet$^+$} as a
function of $m$. Both techniques run in linear time with respect to
$m$.

In summary, the computation time of {\em Privelet$^+$} is linear to $n$
and $m$, which confirms our analysis that {\em Privelet$^+$} has an
$O(n + m)$ time complexity. Compared to {\em Basic}, {\em
  Privelet$^+$} incurs a higher computation overhead, but this is
justified by the facts that {\em Privelet$^+$} provides much
better utility for range-count queries.

\section{Related Work} \label{sec:related}

Numerous techniques have been proposed for ensuring
$\epsilon$-differential privacy in data publishing
\cite{dmn06,kkm09,gmw09,cm08,blr08,kln08,nrs07,mka08,bcd07}. The
majority of these techniques, however, are not designed for the
publication of general relational tables. In particular, the solutions
by Korolova et al.\cite{kkm09} and G\"{o}tz et al.\cite{gmw09} are
developed for releasing query and click histograms from search
logs. Chaudhuri and Monteleoni \cite{cm08}, Blum et al.\cite{blr08},
and Kasiviswanathan et al.\cite{kln08} investigate how the results of
various machine learning algorithms can be published. Nissum et
al.\cite{nrs07} propose techniques for releasing (i) the median value
of a set of real numbers, and (ii) the centers of the clusters output
from the $k$-means clustering algorithm. Machanavajjhala et
al.\cite{mka08} study the publication of {\em commmuting patterns},
i.e., tables with a scheme $\langle${\em ID}, {\em Origin}, {\em
  Destination}$\rangle$ where each tuple captures the residence and
working locations of an individual.

The work closest to ours is by Dwork et al.\cite{dmn06} and Barak et
al.\cite{bcd07}. Dwork et al.'s method, as discussed previously, is
outperformed by our {\em Privelet} technique in terms of the accuracy
of range-count queries. On the other hand, Barak et al.'s technique is
designed for releasing {\em marginals}, i.e., the projections of a
frequency matrix on various subsets of the dimensions. Given a set of
marginals, Barak et al.'s technique first transforms them into the
Fourier domain, then adds noise to the Fourier coefficients. After
that, it refines the noisy coefficients, and maps them back to a set of
noisy marginals. Although this technique and {\em Privelet} have a
similar framework, their optimization goals are drastically
different. Specifically, Barak et al.'s technique does not provide
utility guarantees for range-count queries; instead, it ensures that
(i) every entry in the noisy marginals is a non-negative integer, and
(ii) all marginals are mutually consistent, e.g., the sum of all
entries in a marginal always equals that of another marginal.

\begin{figure}[t]
\begin{minipage}[t]{1.72 in}
\centering
\includegraphics[height=34.5mm]{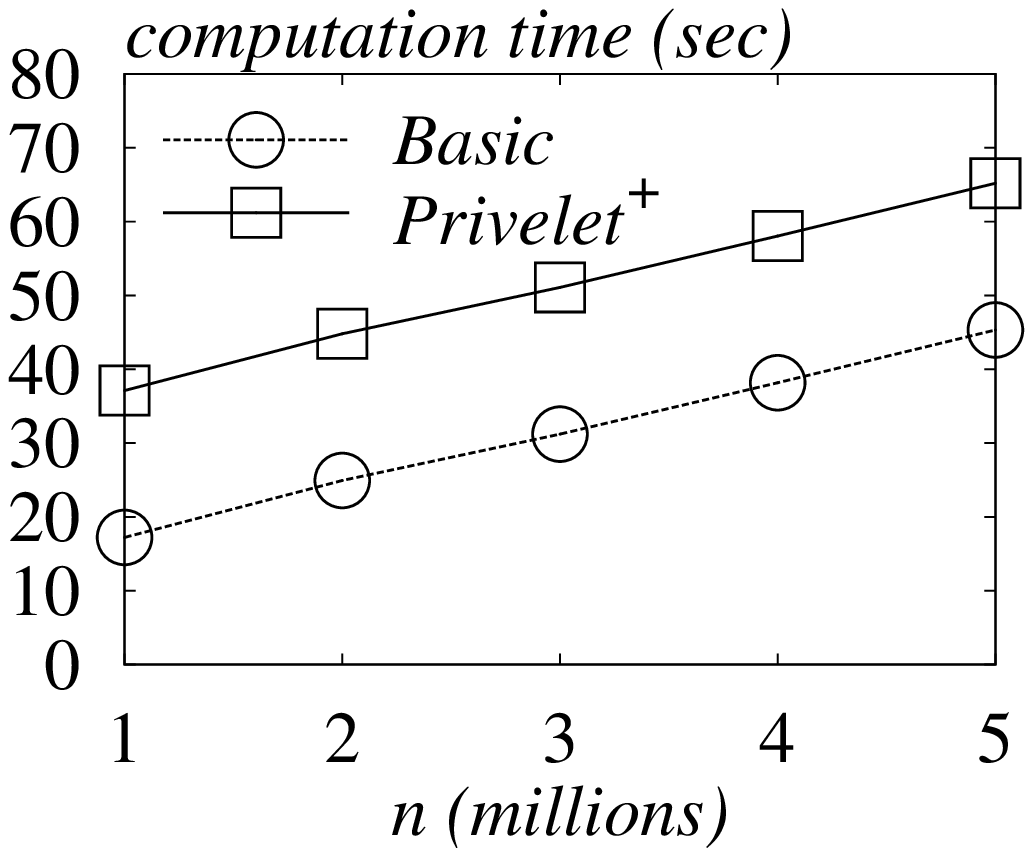}
\figcapup \caption{Computation Time vs.\ $n$} \figcapdown \label{fig:time-vs-n}
\end{minipage}
\begin{minipage}[t]{1.72 in}
\centering
\includegraphics[height=34.5mm]{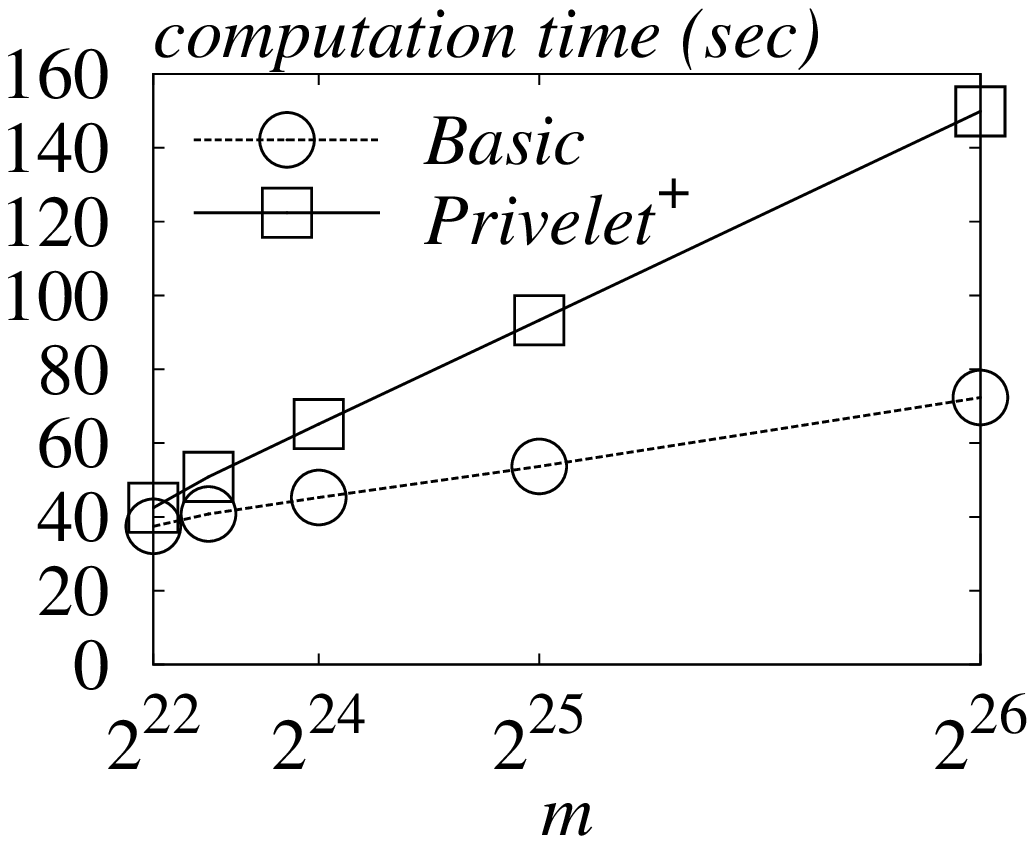}
\figcapup \caption{Computation Time vs.\ $m$} \figcapdown \label{fig:time-vs-m}
\end{minipage}
\end{figure}

In addition, Barak et al's technique requires solving a linear program where the number of variables equals the number $m$ of entries in the frequency matrix. This can be computationally challenging for practical datasets with a large $m$. For instance, for the two census datasets used in our experiments, we have $m > 10^8$. In contrast, {\em Privelet} runs in time linear to $m$ and the number $n$ of tuples in the input table.

Independent of our work, Hay et al.\cite{hrm09} propose an approach for achieving $\epsilon$-differential privacy while ensuring polylogorithmic noise variance in range-count query answers. Given a one-dimensional frequency matrix $M$, Hay et al.'s approach first computes the results of a set of range-count queries on $M$, and then adds Laplace noise to the results. After that, it derives a noisy frequency matrix $M^*$ based on the noisy query answers, during which it carefully exploits the correlations among the answers to reduce the amount of noise in $M^*$. Although Hay et al.'s approach and {\em Privelet} provide comparable utility guarantees, the former is designed exclusively for one-dimensional datasets, whereas the latter is applicable on datasets with arbitrary dimensionalities.

There also exists a large body of literature (e.g.,
\cite{cgr01,vw99,gk05}) on the application of wavelet transforms in
data management. The focus of this line of research, however, is not
on privacy preservation. Instead, existing work mainly
investigates how wavelet transforms can be used to construct space-
and time-efficient representations of multi-dimensional data, so as to
facilitate query optimization \cite{vw99}, or approximate query
processing \cite{cgr01,gk05}, just to name two applications.

\section{Conclusions} \label{sec:conclu}

We have presented {\em Privelet}, a data publishing technique that
utilizes wavelet transforms to ensure $\epsilon$-differential
privacy. Compared to the existing solutions, {\em Privelet} provides
significantly improved theoretical guarantees on the accuracy of
range-count queries. Our experimental evaluation demonstrates the
effectiveness and efficiency of {\em Privelet}.

For future work, we plan to extend {\em Privelet} for the case where
the distribution of range-count queries is known in advance. Furthermore,
currently {\em Privelet} only provides bounds on the noise variance in
the query results; we want to investigate what guarantees {\em
  Privelet} may offer for other utility metrics, such as the expected
relative error of the query answers.

\bibliographystyle{IEEEtran}

\bibliography{ref}

\appendix{}

\subsection{Proof of Lemma~\ref{lmm:over-privacy}} \label{sec:proof-over-privacy}

Let $T_1$ and $T_2$ be any two tables that differ in only one tuple, $M_1$ and $M_2$ be the frequency matrices of $T_1$ and $T_2$, respectively. Let $T_3 = T_1 \cap T_2$, and $M_3$ be the frequency matrix of $T_3$. Observe that $M_1$ and $M_3$ differ in only one entry, and the entry's value in $M_1$ differs from its value in $M_3$ by one. Since $F$ has a generalized sensitivity $\rho$ with respect to $\mathcal{W}$,
\begin{equation} \nonumber
\sum_{f \in F} \Big( \mathcal{W}(f) \cdot |f(M_1) - f(M_3)| \Big) \le \rho \cdot \|M_1 - M_3\|_1 = \rho.
\end{equation}
Similarly, we have
\begin{equation} \nonumber
\sum_{f \in F} \Big( \mathcal{W}(f) \cdot |f(M_2) - f(M_3)| \Big) \le \rho \cdot \|M_2 - M_3\|_1 = \rho.
\end{equation}

Let $f_i$ ($i \in [1, |F|]$) be the $i$-th query in $F$, and $x_i$ be an arbitrary real number. We have\\
\begin{flalign*}
&\frac{Pr\big\{\mathcal{G}(T_2) = \langle x_1, x_2, \ldots, x_{|F|} \rangle\big\}}{Pr\big\{\mathcal{G}(T_2) = \langle x_1, x_2, \ldots, x_{|F|} \rangle\big\}} \\
 & \; =  \frac{\Pi_{i=1}^{|F|} \left(\frac{\mathcal{W}(f_i)}{2 \lambda} \cdot \exp\Big(-\mathcal{W}(f_i) \cdot |x_i- f_i(M_2)| / \lambda\Big)\right)}{\Pi_{i=1}^{|F|} \left(\frac{\mathcal{W}(f_i)}{2 \lambda} \cdot \exp\Big(-\mathcal{W}(f_i) \cdot |x_i - f_i(M_1)| / \lambda\Big)\right)} \\
 & \; \le  \Pi_{i=1}^{|F|} \exp\Big(\mathcal{W}(f_i) \cdot  |f_i(M_1)- f_i(M_2)| / \lambda\Big) \\
  & \; \le  \Pi_{i=1}^{|F|} \exp\Big(\mathcal{W}(f_i) \cdot |f_i(M_1)- f_i(M_3)| / \lambda \\
  & \qquad + \mathcal{W}(f_i) \cdot |f_i(M_2)- f_i(M_3)|/\lambda \Big)  \\
 & \; \le  e^{2 \rho / \lambda},
\end{flalign*}
which completes the proof.

%\subsection{Proof of Lemma~\ref{lmm:ord-amp}} \label{sec:proof-ord-amp}

\subsection{Proof of Lemma~\ref{lmm:ord-utility}} \label{sec:proof-ord-utility}

Let $R$ be the decomposition tree of $M^*$. Recall that each entry $v$
in $M^*$ can be expressed as a weighted sum (see
Equation~\ref{eqn:ord-recon}) of the base coefficient $c_0 \in C$ and
the ancestors of $v$ in $R$. In particular, the base coefficient has a
weight $1$ in the sum. On the other hand, an ancestor $c$ of $v$ has a
weight $1$ ($-1$) in the sum, if $v$ is in the left (right) subtree of
$c$. Therefore, for any one-dimensional range-count query with a
predicate ``$A_1 \in S_1$'', its answer on $M^*$ can be formulated as
a weighted sum $y$ of the wavelet coefficients as follows:
\begin{equation} \nonumber
y = |S_1| \cdot c_0 + \sum_{c \in C \backslash \{c_0\}} \Big( c \cdot \big(\alpha(c) - \beta(c)\big) \Big),
\end{equation}
where $\alpha(c)$ ($\beta(c)$) denotes the number of leaves in the left (right) subtree of $c$ that are contained in $S_1$.

For any coefficient $c$, if none of the leaves under $c$ is contained in $S_1$, we have $\alpha(c) = \beta(c) = 0$. On the other hand, if all leaves under $c$ are covered by $S_1$, then $\alpha(c) = \beta(c) = 2^{l - level(c)}$, where $level(c)$ denotes the level of $c$ in $R$. Therefore, $\alpha(c) - \beta(c) \ne 0$, if and only if the left or right subtree of $c$ partially intersects $S_1$. At any level of the decomposition tree $R$, there exist at most two such coefficients, since $S_1$ is an interval defined on $A_1$.

Let $l = \log_2 |M^*|$. Consider a coefficient $c$ at level $i$ ($i \in [1, l]$) of $R$, such that $\alpha(c) - \beta(c) \ne 0$. Since the left (right) subtree of $c$ contains at most $2^{l-i}$ leaves, we have $\alpha(c), \beta(c) \in [0, 2^{l-i}]$. Therefore, $|\alpha(c) - \beta(c)| \le 2^{l-i}$. Recall that $\mathcal{W}_{Haar}(c) = 2^{l-i+1}$; therefore, the noise in $c$ has a variance at most $\sigma^2/ 4^{l-i+1}$. In that case, the noise contributed by $c$ to $y$ has a variance at most
\begin{eqnarray*}
\big(\alpha(c) - \beta(c)\big)^2 \cdot \sigma^2 / 4^{l-i+1} & \le &  \big(2^{l-i}\big)^2 \cdot \sigma^2 / 4^{l-i+1} \\
& = & \sigma^2/4.
\end{eqnarray*}
On the other hand, the noise in the base coefficient $c_0$ has a variance at most $\left(\sigma/|M^*|\right)^2$. Therefore, the noise contributed by $c_0$ to $y$ has a variance at most $|S_1|^2 \cdot \left(\sigma / |M^*|\right)^2$, which is no more than $\sigma^2$.

In summary, the variance of noise in $y$ is at most
\begin{eqnarray*}
\sigma^2 + 2 \cdot l \cdot \sigma^2 / 4 &=& (2 + \log_2 |M^*|) /2 \cdot \sigma^2,
\end{eqnarray*}
which completes the proof.

\subsection{Proof of Lemma~\ref{lmm:nom-utility}} \label{sec:proof-nom-utility}

We will prove the lemma in two steps: The first step analyzes the variance of noise in each coefficient in $C^*$; the second step shows that the result of any range-count query can be expressed as a weighted sum of the coefficients in $C^*$, such that the variance of the sum is less than $\sigma^2$.

Let $C$ be the set of nominal wavelet coefficients of the input frequency matrix $M$. Let $G$ be any sibling group in $C$, and $G'$ ($G^*$) be the corresponding group in $C'$ ($C^*$). By the way each coefficient in $C$ is computed, $\sum_{g \in G} g = 0$. Let $g_i$, $g'_i$, and $g^*_i$ denote the $i$-th ($i \in [1, |G|]$) coefficient in $G$, $G'$, and $G^*$, respectively. Let $\eta_i = g'_i - g_i$ be the noise in $g'_i$. We have
\begin{align} \label{eqn:nom-utility-1}
g^*_1 \; &= \; g'_1 - \frac{1}{|G|} \cdot \sum_{i = 1}^{|G|} g'_i \; = \; g_1 + \eta_1 - \frac{1}{|G|} \cdot \sum_{i = 1}^{|G|} (g_i + \eta_i) \nonumber \\
%&= g_1 + (1 - \frac{1}{|G|}) \cdot \eta_1 - \frac{1}{|G|} \sum_{i=2}^{|G|} \eta_i + \frac{1}{|G|} \sum_{i = 1}^{|G|} g_i \nonumber \\
&= \; g_1 + (1 - \frac{1}{|G|}) \cdot \eta_1 - \frac{1}{|G|} \sum_{i=2}^{|G|} \eta_i
\end{align}
Recall that $\mathcal{W}_{Nom} (g_i) = 1/(2 - 2/|G|)$. Hence, $\eta_i$ has a variance at most $4 (1 - 1/|G|)^2 \cdot \sigma^2$. By Equation~\ref{eqn:nom-utility-1}, the noise in $g^*_1$ has a variance at most
\begin{align} \label{eqn:nom-utility-2}
&\left((1 - 1/|G|)^2 + \left(1/|G|\right)^2 \cdot (|G| - 1) \right) \cdot 4 (1 - 1/|G|)^2 \cdot \sigma^2 \nonumber \\
& \qquad \qquad \qquad = 4 \left(1 - 1/|G|\right)^3 \cdot \sigma^2.
\end{align}

Similarly, it can be proved that for each non-base coefficient $c^*$ in $C^*$, the variance of the noise in $c^*$ is at most $4 \cdot \left(1 - 1/f\right)^3 \cdot \sigma^2$, where $f$ is the fanout of $c^*$'s parent in the decomposition tree $R$. On the other hand, the base coefficient $c^*_0$ in $C^*$ has a noise variance at most $\sigma^2$, since it is identical to the base coefficient in $C'$.

Now consider any range-count query $q$, such that the predicate in $q$ corresponds to a certain node $N$ in the hierarchy $H$ associated with $M^*$. Let $h$ be the height of $H$. Given $M^*$, we can answer $q$ by summing up the set $S$ of entries that are in the subtree of $N$ in $H$. If $N$ is a leaf in $H$, then $S$ should contain only the entry $v^* \in M^*$ that corresponds to $N$. By Equation~\ref{eqn:nom-recon},
\begin{equation} \label{eqn:nom-utility-3}
v^* = c^*_{h-1} + \sum_{i=0}^{h-2} \left( c^*_i \cdot \prod_{j=i}^{h-2} \frac{1}{f_j} \right),
\end{equation}
where $c^*_i \in C^*$ is the ancestor of $v^*$ at the $(i+1)$-th level of the decomposition tree, and $f_i$ is the fanout of $c^*_i$. As discussed above, $c^*_0$ has a noise variance at most $\sigma^2$, while $c^*_i$ ($i \in [1, h-1]$) has a noise variance at most $4 \cdot \left(1 - 1/f_{i-1}\right)^3 \cdot \sigma^2$. By Equation~\ref{eqn:nom-utility-3} and the fact that $f_i \ge 1$, it can verified that the variances of the noise in $v^*$ is less than $4 \sigma^2$.

On the other hand, if $N$ is a level-($h-1$) node in $H$, then $S$ should contain all entries in $M^*$ that are children of $N$ in $H$. Observe that each of these entries has a distinct parent in the decomposition tree $R$, but they have the same ancestors at levels $1$ to $h-2$ of $R$. Let $c^*_i$ be the ancestor of these entries at level $i+1$, and $f_i$ be the fanout of $c^*_i$. Let $X$ be the set of wavelet coefficients in $R$ that are the parents of the entries in $S$. Then, $X$ should be a sibling group in $C^*$, and $|X| = f_{h-2}$. In addition, we have $\sum_{c^* \in X} c^* = 0$, as ensured by the mean substraction procedure. In that case, by Equation~\ref{eqn:nom-recon}, the sum of the entries in $S$ equals
\begin{align*}
\sum_{v^* \in S} v^* &= \sum_{c^* \in X} c^* + |X| \cdot \sum_{i=0}^{h-2} \left( c^*_i \cdot \prod_{j=i}^{h-2} \frac{1}{f_j} \right) \\
&= c^*_{h-2} + \sum_{i = 0}^{h-3} \left( c^*_i \cdot \prod_{j=i}^{h-3} \frac{1}{f_j} \right).
\end{align*}
Taking into account the noise variance in each $c^*_i$ and the fact that $f_j \ge 1$, we can show that the variance of noise in $\sum_{v^* \in S} v^*$ is also less than $4 \sigma^2$.

In general, we can prove by induction that, when $N$ is a level $k$ ($k \in [1, h-2]$) node in $H$, the answer for $q$ equals
\begin{equation}  \label{eqn:nom-utility-4}
c^*_{k - 1} + \sum_{i = 0}^{k-2} \left( c^*_i \cdot \prod_{j=i}^{k-1} \frac{1}{f_j} \right),
\end{equation}
where $c^*_{k-1}$ is a wavelet coefficient at level $k$ of the decomposition tree $R$, $c^*_i$ ($i \in [0, k-2]$) is the ancestor of $c^*_{k-1}$ at level $i+1$, and $f_i$ is the fanout of $c^*_i$. Based on Equation~\ref{eqn:nom-utility-4}, it can be shown that the noise variance in the answer for $q$ is less than $4 \sigma^2$, which completes the proof.

\subsection{Proof of Theorem~\ref{thrm:multi-amp}} \label{sec:proof-multi-amp}

Let $M'$ be any matrix that can be obtained by changing a certain entry $v$ in the input matrix $M$. Let $\delta = \|M - M'\|_1$, and $C_i$ ($C'_i$) be the step-$i$ matrix in the HN wavelet transform on $M$ ($M'$). For any $j \in [1, m]$, let $C_i(j)$ and $C'_i(j)$ denote the $j$-th entry in $C_i$ and $C'_i$, respectively. By the definition of generalized sensitivity, Theorem~\ref{thrm:multi-amp} holds if and only if the following inequality is valid:
\begin{eqnarray} \label{eqn:multi-amp-1}
\sum_{j = 1}^{m}\Big( \mathcal{W}_{HN}\big( C_d(j) \big) \cdot | C_d(j) - C'_d(j) | \Big) \le \delta \cdot \prod_{i=1}^d \mathcal{P}(A_i).
\end{eqnarray}
To establish Equation~\ref{eqn:multi-amp-1}, it suffices to prove that the following inequality holds for any $k \in [1, d]$.
\begin{eqnarray} \label{eqn:multi-amp-2}
\sum_{j = 1}^{m}\Big( \mathcal{W}_{HN}\big( C_k(j) \big) \cdot | C_k(j) - C'_k(j) | \Big) \le \delta \cdot \prod_{i=1}^k \mathcal{P}(A_i).
\end{eqnarray}
Our proof for Equation~\ref{eqn:multi-amp-2} is based on an induction on $k$. For the base case when $k=1$, Equation~\ref{eqn:multi-amp-2} directly follows from Lemmas \ref{lmm:ord-amp} and \ref{lmm:nom-amp}. Assume that Equation~\ref{eqn:multi-amp-2} holds for some $k=l \in [1, d-1]$. We will show that the case when $k = l+1$ also holds.

Consider that we transform $C_l$ into $C'_l$, by replacing the entries in $C_l$ with the entries in $C'_l$ one by one. The replacement of each entry in $C_l$ would affect some coefficients in $C_{l+1}$. Let $C^{(\alpha)}_{l+1}$ denote the modified version of $C_{l+1}$ after the first $\alpha$ entries in $C_l$ are replaced. By Lemmas \ref{lmm:ord-amp} and \ref{lmm:nom-amp} and by the way we assign weights to the coefficients in $C_{l+1}$,
\begin{flalign*}
\sum_{j=1}^m \Big( \frac{\mathcal{W}_{HN}\big(C_{l+1}(j)\big)}{\mathcal{W}_{HN}\big(C_{l}(\alpha)\big)} \cdot \left|C^{(\alpha)}_{l+1}(j) - C^{(\alpha-1)}_{l+1}(j)\right| \Big) \\
\le \mathcal{P}(A_{l+1}) \cdot \left|C_{l}(\alpha) - C'_l(\alpha)\right|.
\end{flalign*}
This leads to
\begin{align} \label{eqn:multi-amp-3}
&\mathcal{P}(A_{l+1}) \cdot \sum_{\alpha=1}^m \Big( \mathcal{W}_{HN}\big(C_{l}(\alpha)\big) \cdot |C_{l}(\alpha) - C'_l(\alpha)| \Big) \nonumber \\
&\quad \ge \sum_{\alpha = 1}^m \ \sum_{j=1}^m \Big(\mathcal{W}_{HN}\big(C_{l+1}(j)\big) \cdot \left|C^{(\alpha)}_{l+1}(j) - C^{(\alpha-1)}_{l+1}(j)\right| \Big) \nonumber \\
&\quad \ge \sum_{j=1}^m \Big( \mathcal{W}_{HN}\big(C_{l+1}(j)\big) \cdot \left|C_{l+1}(j) - C'_{l+1}(j)\right| \Big)
\end{align}
By Equation \ref{eqn:multi-amp-3} and the induction hypothesis, Equation~\ref{eqn:multi-amp-2} holds for $k = l+1$, which completes the proof.

\subsection{Proof of Theorem~\ref{thrm:multi-utility}} \label{sec:proof-multi-utility}

Our proof of Theorem~\ref{thrm:multi-utility} utilizes the following proposition.

\extraspacing
\begin{proposition} \label{prop:multi-linear}
Let $M$, $M'$, and $M''$ be three matrices that have the same set of dimensions. Let $M_d$, $M'_d$, and $M''_d$ be HN wavelet coefficient matrix of $M$, $M'$, and $M''$, respectively. If $M + M' = M''$, then $M_d + M'_d = M''_d$.
\end{proposition}
\begin{proof} Observe that both the Haar and nominal wavelet transforms are linear transformations, since each wavelet coefficient they produce is a linear combination of the entries in the input matrix. Consequently, the HN wavelet transform, as a composition of the Haar and nominal wavelet transforms, is also a linear transformation. Therefore, $M + M' = M''$ implies $M_d + M'_d = M''_d$.
\end{proof}
\extraspacing

In the following, we will prove Theorem~\ref{thrm:multi-utility} by an induction on $d$. For the base case when $d=1$, the theorem follows directly from Lemmas \ref{lmm:ord-utility} and \ref{lmm:nom-utility}. Assume that theorem also holds for some $d = k \ge 1$. We will prove that the case for $d = k + 1$ also holds.

Let $C^*_k$ be the step-$k$ matrix reconstructed from $C^*_{k+1}$ by applying inverse wavelet transform along the $(k+1)$-th dimension of $C^*_{k+1}$. We divide $C^*_k$ into $|A_{k+1}|$ sub-matrices, such that each matrix $C^*_k[a]$ contains all entries in $C^*_k$ whose last coordinate equals $a \in A_{k+1}$. Observe that each $C^*_k[a]$ can be regarded as a $k$-dimensional HN wavelet coefficient matrix, and can be used to reconstruct a noisy frequency matrix $M^*[a]$ with $k$ dimension $A_1, \ldots, A_k$.

Consider any range-count query $q$ on $M^*$ with a predicate ``$A_i \in S_i$'' on $A_i$ ($i \in [1, k+1]$). Let us define a query $q'$ on each $M^*[a]$, such that $q'$ has a predicate ``$A_i \in S_i$'' for any $i \in [1, k]$. Let $q(M^*)$ be the result of $q$ on $M^*$, and $q'(M^*[a])$ be the result of $q'$ on $M^*[a]$. Let $M' = \sum_{a \in S_{k+1}} M^*[a]$. It can be verified that
\begin{equation} \label{eqn:multi-utility-1}
q(M^*) = \sum_{a \in S_{k+1}} q'(M^*[a]) = q'(M').
\end{equation}
Therefore, the theorem can be proved by showing that the noise in $q'(M')$ has a variance at most $\sigma^2 \cdot \prod_{i=1}^d \mathcal{H}(A_i)$. For this purpose, we will first analyze the the noise contained in the HN wavelet coefficient matrix $C'_k$ of $M'$.

By Proposition~\ref{prop:multi-linear} and the definition of $M'$, we have $C'_k = \sum_{a \in S_{k+1}} C^*_k[a]$. Let $c'$ be an arbitrary coefficient in $M'_k$ with a coordinate $x_i$  on the $i$-th dimension ($i \in [1, k]$). Let $V^*_k$ ($V^*_{k+1}$) be a vector that contains all coefficients in $C^*_k$ ($C^*_{k+1}$) whose coordinates on the first $k$ are identical to those of $c'$. Then, $c'$ can be regarded as the result of a range-count query on $V^*_k$ as follows:
\begin{tabbing}
\hspace{6mm} \=  \kill
\noindent \> \texttt{SELECT} \texttt{COUNT}(*) \texttt{FROM} $V^*_k$ \\
\noindent \> \texttt{WHERE} $A_{k+1} \in S_{k+1}$
\end{tabbing}
Observe that $V^*_k$ can be reconstructed from $V^*_{k+1}$ by applying inverse wavelet transform. Since each coefficient $c^* \in C^*_{k+1}$ has a noise variance $\left(\sigma / \mathcal{W}_{HN}(c^*)\right)^2$, by Lemmas \ref{lmm:ord-utility} and \ref{lmm:nom-utility}, the result of any range-count query on $V^*_k$ should have a noise variance at most
\begin{equation*} \label{eqn:multi-utility-2}
\mathcal{H}(A_{k+1}) \cdot \left( \sigma \cdot \frac{\mathcal{W}_{k+1}(c^*)}{\mathcal{W}_{HN}(c^*)} \right)^2,
\end{equation*}
where $\mathcal{W}_{k+1}$ is the weight function associated with the one-dimensional wavelet transform used to convert $C_k$ into $C_{k+1}$. It can be verified that
\begin{equation*} \label{eqn:multi-utility-3}
\mathcal{W}_{k+1}(c^*) / \mathcal{W}_{HN}(c^*) = 1 / \mathcal{W}_{HN}(c').
\end{equation*}
Therefore, the noise in $c'$ has a variance at most $\mathcal{H}(A_{k+1}) \cdot \left(\sigma / \mathcal{W}_{HN}(c')\right)^2$.

In summary, if we divide the coefficients in $C^*_{k+1}$ into a set $U$ of vectors along the $(k+1)$-th dimension, then each coefficient $c' \in M'_k$ can be expressed as a weighted sum of the coefficients in a distinct vector in $U$. In addition, the weighted sum has a noise with a variance at most $\mathcal{H}(A_{k+1}) \cdot \left(\sigma / \mathcal{W}_{HN}(c')\right)^2$. Furthermore, the noise in different coefficients in $M'_k$ are independent, because (i) the noise in different vectors in $U$ are independent, and (ii) no two coefficients in $M'_k$ correspond to the same vector in $U$. Therefore, by the induction hypothesis, for any range-count query on the $k$-dimensional frequency matrix $M'$ reconstructed from $M'_k$, the query results has a noise with a variance no more than $\sigma^2 \cdot \sum_{i = 1}^{k+1} \mathcal{H}(A_i)$. This, by Equation~\ref{eqn:multi-utility-1}, shows that the theorem also holds for $d = k + 1$, which completes the proof. 

\end{document}